\documentclass{elsarticle}

\usepackage{lineno,hyperref}
\modulolinenumbers[5]

\usepackage{amsmath}
\usepackage{amsfonts}
\usepackage{amssymb}
\usepackage{amsthm}
\usepackage{algorithm}
\usepackage{algcompatible}
\usepackage{tikz}

\newcommand{\tp}[1]{{#1}^{\mathsf T}}

\newcommand{\eps}{\epsilon}
\newcommand{\pa}{\partial}

\renewcommand{\eps}{\varepsilon}
\renewcommand{\epsilon}{\varepsilon}
\renewcommand{\Sigma}{\varSigma}

\newcommand{\Co}{\mathcal{C}}

\newcommand{\Hi}{\mathbb{X}}

\makeatletter 
 
\@addtoreset{algorithm}{section} 
\makeatother

\newtheorem{thm}[algorithm]{Theorem}

\newtheorem{cor}[algorithm]{Corollary}

\newtheorem{asump}[algorithm]{Assumption}

\newtheorem{rk}[algorithm]{Remark}
\newcommand{\E}{\mathrm E}

\newcommand{\tr}{\mathrm{tr}}
\DeclareMathAlphabet\mathbfcal{OMS}{cmsy}{b}{n}

\newcommand{\bn}{ {\bf n}}
\newcommand{\bu}{ {\bf u}}

\newcommand{\bZero}{ {\bf 0}}
\newcommand{\bsigma}{ \boldsymbol{\sigma}}
\renewcommand{\div}{{\rm div}}

\ifpdf
    \graphicspath{{./figure/PNG/}{./figure/PDF/}{./figure/}}
\else
    \graphicspath{{./figure/EPS/}{./figure/}}
\fi

\begin{document}

\begin{frontmatter}

\title{Geometric MCMC for Infinite-Dimensional Inverse Problems}



\author[UCL]{Alexandros Beskos}
\ead{a.beskos@ucl.ac.uk}

\author[Uwarwickstats,ATI]{Mark Girolami}
\ead{M.Girolami@warwick.ac.uk}

\author[CIT]{Shiwei Lan\corref{cor}}
\ead{slan@caltech.edu}

\author[Oxford,Simula]{Patrick E. Farrell}
\ead{patrick.farrell@maths.ox.ac.uk}

\author[CIT]{Andrew M. Stuart\corref{cor}}
\ead{A.M.Stuart@warwick.ac.uk}

\cortext[cor]{Corresponding author}

\address[UCL]{Department of Statistical Science, University College London, Gower Street, London, WC1E 6BT, UK}
\address[Uwarwickstats]{Department of Statistics, University of Warwick, Coventry CV4 7AL, UK}
\address[ATI]{The Alan Turing Institute for Data Science, British Library, 96 Euston Road, London, NW1 2DB, UK}
\address[CIT]{Department of Computing + Mathematical Sciences, California Institute of Technology, Pasadena, CA 91125, USA}
\address[Oxford]{Mathematical Institute, University of Oxford, Andrew Wiles Building, Radcliffe Observatory Quarter, Woodstock Road, Oxford, OX2 6GG, UK}
\address[Simula]{Center for Biomedical Computing, Simula Research Laboratory, Martin Linges vei 17, 1364 Fornebu, Norway.}

\begin{abstract}
Bayesian inverse problems often involve sampling posterior distributions on infinite-dimensional function spaces. Traditional Markov chain Monte Carlo (MCMC) algorithms are characterized by deteriorating mixing times upon mesh-refinement, when the finite-dimensional 
approximations become more accurate.
Such methods are typically forced to reduce step-sizes as the discretization gets finer, and thus are expensive as a function 
of dimension. Recently, a new class of MCMC methods with mesh-independent convergence times has emerged.
However, few of them take into account the geometry of the posterior informed by the data. 
At the same time, recently developed geometric MCMC algorithms have been 
found to be powerful in exploring complicated distributions that deviate 
significantly from elliptic Gaussian laws, but are in general computationally
intractable for models defined in infinite dimensions. In this work, we 
combine geometric methods on a finite-dimensional subspace with 
mesh-independent infinite-dimensional approaches. Our objective is to speed 
up MCMC mixing times, without significantly increasing the computational 
cost per step (for instance, in comparison with the vanilla preconditioned Crank-Nicolson 
(pCN) method). This is achieved by using ideas from geometric MCMC to
probe the complex structure of an intrinsic finite-dimensional subspace 
where most data information concentrates, while retaining robust mixing 
times as the dimension grows by using pCN-like methods in the
complementary subspace. The resulting algorithms are demonstrated in
the context of three challenging inverse problems arising in subsurface flow,
heat conduction and incompressible flow control. The algorithms
exhibit up to two orders of magnitude improvement in sampling efficiency 
when compared with the pCN method.
\end{abstract}


\end{frontmatter}

\linenumbers


\section{Introduction}
In this work we consider Bayesian inverse problems where the objective is to identify an unknown 
function 
parameter $u$ which is an element of a separable Hilbert space $(\Hi, \langle\cdot, \cdot \rangle,\lvert\cdot \rvert)$. 
All probability measures on $\Hi$ in the rest of the paper are assumed to be defined 
on the standard Borel $\sigma$-algebra $\mathcal{B}(\Hi)$.
We are given finite-dimensional observations $y\in \mathbb{Y}=\mathbb{R}^{m}$, for $m\ge 1$, with
$u$ and $y$ being connected via the mapping:
\begin{equation}\label{eq:forward}
y=\mathcal G(u) + \eta\ ,\quad \eta\sim f \ , 
\end{equation}
for some noise distribution $f$, with $u$ representing the unknown parameter 
of a (non-linear) PDE and $\mathcal G:\Hi\mapsto \mathbb{Y}$ the related forward solution operator
for the PDE  mapping $u$ onto the data space~$\mathbb{Y}$.
In a Bayesian setting, a prior measure $\mu_0$ is assigned to $u$. 
With a small abuse of notation, we denote also by $f$ the density (assumed to exist) of the noise distribution
with respect to the Lebesgue measure, thus we define the negative log-likelihood $\Phi:\Hi\times \mathbb{Y}\to \mathbb R$ as:
\begin{equation*}
\Phi(u;y)=-\log f\big\{\big(y-\mathcal{G}(u)\big);u\big\} \ , 
\end{equation*}
with $f\{\cdot\,;u\}$ indicating the density function for a given $u$.
Denoting by $\mu^y$ the posterior of $u$, and using Bayes' theorem, 
we get:
\begin{equation*}
\frac{d\mu^y}{d\mu_0}(u) = \frac{1}{Z}\,\exp(-\Phi(u;y))
\end{equation*}
for a normalising constant $Z=\int_{\Hi} \exp(-\Phi(u;y)) \mu_0(du)$ assumed positive and finite.

In this work we consider a Gaussian prior $\mu_0 = \mathcal N(0,\mathcal C)$ with the covariance $\mathcal C$ being a positive, self-adjoint and trace-class operator on $\Hi$.
Notice that the posterior $\mu^y$ can exhibit strongly non-Gaussian behaviour,
with finite-dimensional projections having complex non-elliptic contours, although the existence of a density with respect to $\mu_0$ does imply near-Gaussianity for appropriate tail components of the target law $\mu^{y}$. 

Sampling from $\mu^y$ in the context of PDE-constrained inverse problems is typically a very challenging 
undertaking due to the high-dimensionality of 
the target, the  non-Gaussianity   of   the posterior and the computational burden of repeated PDE solutions for evaluating the likelihood function at different   parameters. 
It is now well-understood that traditional Metropolis-Hastings algorithms
have deteriorating mixing times upon refinement of the mesh-size used in practice in  
the finite-dimensional projection of parameter $u$.
This has prompted the recent development of a class of `advanced' MCMC methods that avoid this deficiency,
see for instance the line of works in \cite{beskos08, beskos11, cotter13,
law2014, pinski15, rudolf15, cui16}.
The main difference of the new methodology compared to standard Metropolis-Hastings 
is that the algorithms are well-defined on the infinite-dimensional Hilbert space. This yields the important computational benefit of mesh-independent mixing times for the practical finite-dimensional algorithms ran on the computer.


This work makes a number of contributions.
First, we generalize geometric MCMC methods - the simplified Riemannian manifold Metropolis-adjusted Langevin algorithm (MALA) of \cite{girolami11} and a Hamiltonian Monte-Carlo (HMC) extension of it - from finite to infinite dimensions. 
Unlike recent development of geometric methods including Stochastic Newton (SN) MCMC \citep{martin12} and Riemannian manifold Hamiltonian Monte Carlo for large-scale PDE-constrained inverse problems \citep{bui12},
these proposed advanced MCMC algorithms are well-defined on the Hilbert space. They have the capacity to both explore complex probability structures and have robust mixing times in high dimensions.
Our methodology can also be thought of as a generalisation of the operator-weighted proposal of \cite{law2014} or the dimension-independent likelihood 
informed (DILI) MCMC method of \cite{cui16} which exploit the posterior curvature at a fixed point obtained via an optimiser  or through adaptive  averaging over samples; our methodology invokes position dependent
curvatures to allow for more flexible geometric adaptation. We provide high-level conditions and rigorous proofs for the well-posedness of the new methods on infinite-dimensional Hilbert spaces.
Second, we establish connections between MALA- and HMC-type algorithms in the infinite dimensional setting.
HMC algorithms, viewed as multi-step generalizations of their MALA analogues, make big jumps that suppress random-walk behavior
and can provide numerical advantages over MALA 
by substantially  reducing mixing times.
Third, we develop a straightforward dimension reduction methodology which renders the methods highly
effective from a practical viewpoint. 
Our methods aim to adapt to the  local curvature of the target and provide proposals 
which are appropriate for non-linear likelihood-informed subspaces.  
A simpler step is then developed for a complementary subspace obtained 
by truncating the Karhunen-Lo\`eve expansion of the Gaussian prior.
Other such separation methods used in the non-geometric context 
(likelihood informed subspace \citep[LIS][]{cui14} or the active subspace \citep[AS][]{constantine15a}) could potentially be brought into our setting, though this requires further research. 
Lastly, we apply the geometric methods together with other main MCMC algorithms on three challenging inverse problems and contrast their efficiency.
Two elliptic inverse problems, involving a groundwater flow and a thermal fin, aim to infer the coefficients of the elliptic PDEs (representing the permeability of a porous medium and the heat conductivity of a material respectively) from data taken at given locations of the forward solver. The third inverse problem involves an incompressible Navier-Stokes equation, with the objective to infer the inflow velocity given sparse observations from the downstream outlet boundary. To the best of our knowledge, it is the first successful application of geometric MCMC methods to non-linear infinite dimensional
inverse problems and demonstration of their effectiveness in this field.
We should mention here that an important paper in this context is \cite{martin12} which
introduced the Stochastic Newton (SN) method. Although the derivation
of the algorithm was not infinite-dimensional, the authors do show
that on linear Gaussian problems the acceptance probability is one,
an essential ingredient in the definition of an infinite-dimensional
sampler.  We also mention that the paper \cite{petra14} generalizes 
the SN method by considering variants in which the Hessian
is frozen at the maximum a posteriori (MAP) estimator, and low-rank approximations are
employed; the methodology is applied to a non-linear ice sheet inverse problem
with considerable success.
The SN algorithm of  \cite{martin12} can be identified as a special case of our scheme and further details are given in Subsection \ref{sec:mMALA}).

The paper is organized as follows.
Section \ref{sec:MCMC} reviews the recently introduced MCMC methods on infinite-dimensional Hilbert spaces.
Section \ref{sec:geo} develops the new geometric 
MCMC methods and establishes their well-posedness under certain conditions.
Section \ref{sec:numerics} applies the new methodology to a number of complex inverse problems 
and shows that use of information about the underlying geometry can provide 
significant computational improvements in the cost per unit sample.
Section \ref{sec:conclusion} concludes with a summary 
and a suggested path for several future investigations. 


\section{(Non-Geometric) MCMC on Hilbert Spaces}
\label{sec:MCMC}
We review some of the advanced MCMC methods published in the 
literature, see e.g.\@ \cite{beskos08, beskos11, cotter13} or \cite{cui16} 
for recent contributions.
 For simplicity we drop  $y$ from the various terms involved, so we denote the posterior as $\mu(du)$ and the potential function as $\Phi(u)$.
For target $\mu(du)$ and the various proposal kernels $Q(u,du')$ in the sequel, we define the bivariate 
law: 
\begin{equation}
\label{eq:MH1}
\nu(du,du') = \mu(du)\,Q(u,du')\ . 
\end{equation}
Following the theory of Metropolis-Hastings on general spaces \citep{tierney98},
the acceptance probability 
$a(u,u')$ is non-trivial when $\nu(du,du')\simeq\nu^{\top}(du,du')$ 
with $\nu^{\top}$ denoting the symmetrisation of $\nu$, that is 
\begin{equation}
\label{eq:MH2}
\nu^{\top}(du,du'):=\nu(du',du)\  .
\end{equation}
The symbol $(\simeq)$ denotes 
absolute continuity between probability measures.
The acceptance probability is then:
\begin{equation}
\label{eq:MH3} 
a(u,u') = 1\wedge \frac{d\nu^{\top}}{d\nu}(u,u')\ . 
\end{equation}
where $\alpha \wedge \beta$ denotes the minimum of $\alpha, \beta\in \mathbb{R}$.

The preconditioned Crank-Nicolson (pCN) method \citep{neal10,beskos08,cotter13} is a modification of the standard random-walk Metropolis (RWM). The method is described in Algorithm \ref{tab:SMC} 
and involves a  free parameter $\rho\in [0,1)$ controlling the size of move from 
the current position.
\begin{algorithm}[!ht]
\begin{flushleft}
\medskip
{\itshape
\begin{enumerate}
\item[\textit{1.}] Given current $u$, sample independently $\xi \sim \mathcal N(0, \mathcal C)$ and 
propose: $$u'=\rho\,u+\sqrt{1-\rho^2}\,\xi\ . $$
\item[\textit{2.}]  Accept $u'$ with probability $1\wedge\exp\big\{-\Phi(u')+\Phi(u)
\big\}$, 
otherwise stay at $u$.
%
%
\end{enumerate} }
\medskip
\end{flushleft}
\vspace{-0.4cm}
\caption{A single Markov step for pCN.}
\label{tab:SMC}
\end{algorithm}
PCN is well-defined on the Hilbert space $\Hi$ with the proposal being prior-preserving, whereas standard RWM 
can only be defined on finite-dimensional discretization and has diminishing 
acceptance probability for fixed step-size and increasing resolution 
\citep{roberts97}. 
Thus, pCN mixes faster than RWM in high-enough dimensions and the disparity in mixing rates becomes greater upon mesh-refinement \citep{cotter13}.
However, pCN in general does not use the data in the proposal and can exhibit strong diffusive behavior when exploring complex posteriors.
We note here that some recent contributions 
\citep{law2014, pinski15, rudolf15} aim to adapt 
the pCN proposal to the covariance structure
of the target.

One approach for developing data-informed methods is to take advantage of gradient information in a steepest-descent setting.
Consider the Langevin SDE on the Hilbert space, preconditioned by some operator $K$:
\begin{equation}\label{eq:Langevin}
\frac{du}{dt} = -\frac12\,K\,\big\{ \Co^{-1}u+ D\Phi(u)\big\} + \sqrt{K}\, \frac{dW}{dt}
\end{equation}
with $D\Phi(u)$ denoting the Fr\'echet derivative of $\Phi$ 
(or the corresponding element of the relevant dual space; we will be more precise when defining our new methods in the section \ref{sec:geo}) and $W$ being the cylindrical Wiener process. 
We consider these dynamics under the setting $K=\Co$, when scales are tuned 
to the prior. 
Formally,  SDE \eqref{eq:Langevin} preserves the posterior $\mu$ and  
can be used as the basis for developing effective MCMC proposals \citep{beskos08,cotter13}.
\cite{beskos08} use the following semi-implicit Euler scheme to discretize the above SDE:
\begin{equation}
\label{eq:semi-implicit}
\frac{u'-u}{h} = -\frac12\,\big\{\frac{u+u'}{2} +  \alpha \,\mathcal{C}D\Phi(u)\big\} + \sqrt{\frac{1}{h}}\,\xi\ ,\quad \xi\sim \mathcal N(0,\mathcal C)\ , 
\end{equation}
for an algorithmic parameter $\alpha\equiv 1$ and some small step-size $h>0$. This can be rewritten as:
\begin{equation}
\label{eq:infMALA}
\begin{aligned}
u'&=\rho\,u + \sqrt{1-\rho^2}\,v\ , \quad v= \xi-\tfrac{\alpha\sqrt{h}}{2}\,\mathcal{C} D\Phi(u) \ , \quad 
\rho = (1-\tfrac{h}{4})/(1+\tfrac{h}{4})\ . 
\end{aligned}
\end{equation}
 %
%
%
%
Note that the image space $\mathrm{Im}(\Co^{\frac12})$ is comprised of all $u\in \Hi$ such that  $ \mathcal N(u,\Co) \simeq \mathcal N(0,\Co)$,
see e.g.\@ 
\cite{da14}.
Thus, following \cite{beskos08}, under the assumption that $\Co D \Phi(u)\in 
\mathrm{Im}(\Co^{1/2})$, $\mu_0$-a.s.\@ in~$u$, one can use Theorem 2.21 of \cite{da14} on translations 
of Gaussian measures on separable Hilbert spaces, to obtain 
the following Radon-Nikodym derivative (we denote by  $Q(u,du')$ and $Q_0(u,du')$ the proposal
kernels determined by (\ref{eq:infMALA}) for $\alpha=1$ and $\alpha=0$, respectively):
\begin{equation}
\label{eq:refG}
\frac{dQ(u,\cdot)}{dQ_0(u,\cdot)}(u') = \exp\big\{ -\tfrac{h}{8}\,| \Co^{1/2}D\Phi(u)|^2 - 
\tfrac{\sqrt{h}}{2}\langle D\Phi(u), \tfrac{u'-\rho u}{\sqrt{1-\rho^2}} \rangle     \big\}\  .
\end{equation}
The bivariate Gaussian law 
$\nu_0(du,du'):=\mu_0(du)Q_0(u,du')$
is symmetric ($\nu_0=\nu_0^{\top}$), thus 
one can obtain the Metropolis-Hastings ratio
in the accept/reject (\ref{eq:MH3}) as $d\nu^{\top}/d\nu = (d\nu^{\top}/d\nu_0^{\top})/
(d\nu/d\nu_0)$.
The complete method, labeled $\infty$-MALA (infinite-dimensional MALA), is defined in Algorithm \ref{tab:MALA}.
%

%
\begin{algorithm}[h]
\begin{flushleft}
\medskip
{\itshape
\begin{enumerate}
\item[\textit{1.}] Given current $u$, sample independently $\xi \sim \mathcal N(0, \mathcal C)$ and 
propose: $$u'=\rho\,u + \sqrt{1-\rho^2}\,\big\{ \xi-\tfrac{\sqrt{h}}{2}\,\mathcal{C} D\Phi(u)\}$$
\item[\textit{2.}]  Accept $u'$ with probability  $a(u,u') = 1\wedge \frac{\kappa(u',u)}{\kappa(u,u')}$, where we have set:
\begin{align*}
\kappa(u,u') = \frac{1}{Z}\,\exp\{-\Phi(u)\}\times
 \exp\big\{ -\tfrac{h}{8}\,| \Co^{1/2}D\Phi(u)|^2 - 
\tfrac{\sqrt{h}}{2}\langle D\Phi(u), \tfrac{u'-\rho u}{\sqrt{1-\rho^2}} \rangle     \big\}\  
\end{align*}
%
otherwise stay at $u$.
%
%
\end{enumerate} }
\medskip
\end{flushleft}
\vspace{-0.4cm}
\caption{A single Markov step for $\infty$-MALA.}
\label{tab:MALA}
\end{algorithm}
%


Another likelihood-informed Metropolis-Hastings method involves exploiting Hamiltonian dynamics.
Consider the Hamiltonian differential equation with mass matrix
\footnote{The terminology  `mass matrix' used in 
Hamiltonian dynamical systems should not be confused with the same term used in finite element methods for PDEs.}
equal to $K^{-1}$, that is:
\begin{equation}\label{Hamiltonian}
\frac{d^2u}{dt^2} + 
K \big\{\Co^{-1}u+ D\Phi(u)\big\} = 0
\ .
\end{equation}
These dynamics, considered on the phase-space of
$(u,v)$, for the velocity $v=du/dt$,
preserve the total energy:
\begin{equation*}
H(u,v)=  \Phi(u) + \tfrac{1}{2}\,\langle v, K^{-1} v \rangle\ . 
\end{equation*} 
From a probabilistic point of view, 
when initialized with $v\sim \mathcal N(0, K)$, the 
Hamiltonian dynamics (formally) preserve the target  measure $\mu$ for any integration time, 
and thus they can form the basis for an MCMC method, termed Hybrid (or Hamiltonian) Monte-Carlo (HMC) \citep{duane87,neal10}.
\cite{beskos11} modify the standard HMC algorithm to develop an advanced method that is well-defined  on 
the Hilbert space $\Hi$. We label this algorithm
$\infty$-HMC (infinite-dimensional HMC).
In more detail, setting again $K=\Co$ 
the dynamics in \eqref{Hamiltonian} can be written in the standard form:
\begin{equation}
\label{HD}
\frac{du}{dt} = v\ , 
\quad \frac{dv}{dt} = -u - \Co D\Phi(u)\  .
\end{equation}
Equation  \eqref{HD} gives rise to a  semigroup 
that maps $(u(0),v(0))\mapsto(u(t),v(t))$ and preserves the product measure  $\mu\otimes \mu_0$ under regularity conditions on $\Co$ and $D\Phi(u)$ \citep{beskos11}. Standard HMC 
synthesizes Euler steps on the 
two differential equations in (\ref{HD})
to produce an approximate symplectic integrator.
In contrast, $\infty$-HMC makes use 
of the Strang splitting scheme: 
\begin{align}
\label{eq:split}
du/dt = v\ , 
\quad & dv/dt = -u \ ; \\[0.2cm]
\label{eq:split1}
du/dt = 0\ , 
\quad  
& dv/dt = 
- \Co D\Phi(u)\  ,
\end{align}
and develops a  St\"ormer-Verlet-type
integrator \citep{verlet67,neal10}
by synthesizing solvers of (\ref{eq:split}), 
(\ref{eq:split1}) as follows, for some small $\epsilon>0$ and initial values $(u_0,v_0)$:
\begin{equation}\label{HDdiscret}
\begin{aligned}
v^{-} &= v_0 - \tfrac{\epsilon}{2}\,\Co D\Phi(u_0)\  ;\\
\begin{bmatrix} u_\epsilon\\ v^{+}\end{bmatrix} &= \begin{bmatrix} \cos \epsilon& \sin \epsilon\\ -\sin \epsilon & \cos \epsilon\end{bmatrix}  \begin{bmatrix} u_0\\ v^{-}\end{bmatrix} \ ;\\
v_\epsilon &= v^{+} - \tfrac{\epsilon}{2}\, \Co D\Phi(u_\epsilon)\   .
\end{aligned}
\end{equation}
This scheme, referred to as a leapfrog step,  gives rise to a  map $\Psi_\epsilon: (u_0, v_0)\mapsto (u_\epsilon,v_\epsilon)$.
The algorithm proposes big jumps in the state space by synthesizing  
$I = \lfloor \tau/\epsilon \rfloor$ leapfrog maps, for some time horizon  $\tau>0$. 
It can be shown that if $I=1$ then $\infty$-HMC coincides with $\infty$-MALA for particular choice of step-sizes (see more details in Subsection \ref{sec:inf-mHMC}).
$\infty$-HMC will many times manifest numerical advantages over $\infty$-MALA 
due to the longer, designated moves suppressing random walk behavior.
$\infty$-HMC develops as shown in Algorithm \ref{tab:HMC}, 
where for starting position and velocity $(u,v)$ we have set 
$(u_i,v_i)= \Psi^{i}_{\epsilon}(u,v)$,
with $\Psi_{\epsilon}^{i}$ denoting the synthesis 
of $i$ maps $\Psi_\epsilon$, $0\le i\le I$.
Also, we denote by $\mathcal{P}_u$ the projection onto the $u$-argument.
The derivation of the accept/reject rule is more involved than $\infty$-MALA, 
and requires again that $\Co D \Phi(u)\in 
\mathrm{Im}(\Co^{1/2})$, $\mu_0$-a.s.\@ in $u$;
we refer the reader to \cite{beskos11}. We
will provide full details on the accept/reject
when developing the more general geometric version of $\infty$-HMC in Subsection \ref{sec:inf-mHMC}.
\begin{algorithm}[h]
\begin{flushleft}
\medskip
{\itshape
\begin{enumerate}
\item[\textit{1.}] Given current $u$, sample independently  $v\sim \mathcal N(0,\Co)$
and propose 
$u' = \mathcal{P}_u\big\{ \Psi_\epsilon^{I}( u,v )\big\}$. 
\item[\textit{2.}]  Accept $u'$ with probability $1\wedge \exp\big\{-\Delta H(u,v)\}$
where we have set:
\begin{align*}
\Delta H(u,v) =& \,H(\Psi_\epsilon^I(u,v)) - H(u,v)\\[0.2cm]
& \equiv  \,\Phi(u_I)-\Phi(u_0)-\tfrac{\epsilon^2}{8}\big\{|\Co^{\frac12}D\Phi(u_I)|^2-
|\Co^{\frac12}D\Phi(u_0)|^2\big\}\\
&  \quad  
- 
\tfrac{\epsilon}{2}\, \sum_{i=0}^{I-1}\big(\,\langle v_{i}, D\Phi(u_{i})\rangle + \langle v_{i+1}, D\Phi(u_{i+1})\rangle \,\big)
\end{align*}
%
otherwise stay at $u$.
%
%
\end{enumerate} }
\medskip
\end{flushleft}
\vspace{-0.4cm}
\caption{A single Markov step for $\infty$-HMC.}
\label{tab:HMC}
\end{algorithm}



\section{Geometric Metropolis-Hastings Algorithms}
\label{sec:geo}
Recall the assumed distribution of the data in (\ref{eq:forward}). 
We will be more explicit here and for expository convenience assume Gaussian noise $\eta 
\sim \mathcal{N}_m(0,\Sigma)$, for some symmetric, positive-definite 
$\Sigma\in \mathbb{R}^{m\times m}$. Thus the target distribution is:
\begin{equation*}
\frac{d\mu}{d\mu_0}(u) = \frac{1}{Z}\,\exp(-\Phi(u)) = \frac{1}{Z}\,\exp\big\{\,-\tfrac{1}{2}\,\big|y-\mathcal{G}(u)\big|^{2}_{\Sigma}\,\big\}
\end{equation*}
for some constant $Z>0$, where we have considered the scaled inner product 
$\langle \cdot,\cdot \rangle_{\Sigma} = \langle \cdot, \Sigma^{-1}\cdot \rangle$.
Below, we will define MCMC algorithms on the Hilbert space $\Hi$, and express 
conditions for their well-posedness in terms of the properties of the forward map 
$\mathcal{G}=(\mathcal{G}_k)_{k=1}^{m}:\Hi \mapsto \mathbb{R}^{m}$ which involves regularity properties of the underlying PDE in the given inverse problem. 

We work with the eigenvectors and eigenvalues of the prior covariance  
operator $\Co$, so that $\{\phi_j\}_{j\ge 1}$ is an orthonormal basis of $\Hi$ and
$\{\lambda_j^2\}_{j\ge 1}$ a sequence of positive reals with $\sum{\lambda_j^2}<\infty$ (this enforces 
the trace-class condition for $\Co$),
such that $\Co \phi_j = \lambda_j^2 \phi_j$, $j\ge 1$. 
We make the usual correspondence  
between an element $u$ and its coordinates w.r.t.\@ the basis 
$\{\phi_j\}_{j\ge 1}$, that is
$u=\sum_{j}u_j \phi_j \leftrightarrow \{u_j\}_{j\ge 1}$.
Using the standard Karhunen-Lo\`eve expansion 
of a Gaussian measure \citep{adler10,bogachev98,dashti15} we have the representation:
\begin{equation}
\label{eq:KL}
u \sim \mathcal{N}(0,\Co)\,\,\, \Longleftrightarrow\,\,\,
u=\sum_{j=1}^{\infty} u_j \phi_j\ ,\,\, 
u_j\sim \mathcal{N}(0,\lambda_j^2)\ , \,\,\textrm{ind. over $j\ge 1$}\ .
\end{equation}
We define the Sobolev spaces 
corresponding to the basis $\{\phi_j\}$:
\begin{equation*}
\Hi^s = \big\{\{u_j\}_{j\ge 1}: \sum j^{2s}|u_j|^{2}<\infty  \big\}\ , \quad s\in \mathbb{R}\ ,
\end{equation*}
so that  $\Hi^0\equiv \Hi$ and $\Hi^{s}\subset \Hi^{s'}$ if $s'<s$.
Typically, we will have $\lambda_j = \Theta(j^{-\kappa})$ for some $\kappa>1/2$ 
in the sense that $C_1\cdot j^{-\kappa}\leq 
\lambda_j\leq C_2 \cdot j^{-\kappa}$ for all 
$j \ge 1$, for constants $C_1, C_2>0$.
Thus, the prior (so also the posterior) concentrate on $\Hi^s$ for any $s<\kappa-1/2$.
Notice also that:

$$\mathrm{Im}(\Co^{1/2})= 
\Hi^{\kappa}\ .$$
Assumption \ref{ass:10} imposes some conditions on the gradient $D\Phi(u)$.

\begin{asump}
\label{ass:10}
(i) $\lambda_j = \Theta(j^{-\kappa})$, for $\kappa>1/2$. \\
(ii) For some $\ell\in[0,\kappa-1/2)$, 
the maps $\mathcal{G}_k:\Hi^{\ell}\mapsto \mathbb{R}$, $1\le k \le m$, are Fr\'echet differentiable  
on $\Hi^{\ell}$ with derivatives $D\mathcal{G}_k\in \Hi^{-\ell}$.
\end{asump}

\noindent 
We can assume that $\ell$ is arbitrarily close to $\kappa-1/2$. We make the standard correspondence between 
the bounded linear operator $D\mathcal{G}_k(u)$ on $\Hi^{\ell}$ and an element of its dual space $D\mathcal{G}_k(u)\in \Hi^{-\ell}$
so that $D\mathcal{G}_k(u)(v) = \langle D\mathcal{G}_k(u), v \rangle $ for all $u, v\in\Hi^{\ell}$.
We consider the derivative  $D\mathcal{G}(u) = (D\mathcal{G}_1(u),\ldots 
D\mathcal{G}_m(u) )\in \{\Hi^{-\ell}\}^{m}$, $u\in \Hi^{\ell}$.
Under Assumption \ref{ass:10}, mapping $\Phi$ is Fr\'echet differentiable  
on $\Hi^{\ell}$ with:
\begin{equation*}
D\Phi(u) =  D\mathcal{G}(u)\Sigma^{-1}(\mathcal{G}(u) - y)\in \Hi^{-\ell}\ . 
\end{equation*}

\subsection{Local Gaussian Approximation of Posterior}

All three MCMC algorithms shown in Section \ref{sec:MCMC}
adjust scales in the proposal according to the prior covariance $\Co$. Indeed, if the target 
distribution was simply $\mu_0$, the proposal dynamics would equalise all scales and would also have acceptance probability equal to~$1$.
However, one can get more effective algorithms if the geometry of the posterior itself 
is taken into consideration in the selection of step-sizes.
We explore in this paper the idea of using a preconditioner $K=K(u)$ which will be location-specific in order to construct algorithms that 
are tuned to the local curvature of the posterior 
as pioneered in \cite{girolami11}, and developed
subsequently in other works, see e.g.\@ \cite{lan14}.

Reviewing $\infty$-MALA 
and $\infty$-HMC methods presented in Section~\ref{sec:MCMC},
the effect of the implicit method (\ref{eq:semi-implicit}) and the splitting (\ref{eq:split}) used for 
$\infty$-MALA and $\infty$-HMC respectively   
is that the resulting scheme provides an `ideal' proposal of acceptance 
probability 1 (respectively of the step-sizes 
$h$ or $\epsilon$) 
for the reference Gaussian measure $\mu_0 = \mathcal{N}(0,\Co)$.
Thinking about the local-move $\infty$-MALA algorithm, if the negative log-density w.r.t.\@ $\mu_0$, $u\mapsto\Phi(u)$, is relatively flat locally around the current position $u$, then one can expect relatively high 
acceptance probability when proposing a move from $u$ for the target $\mu$ itself, for a small 
step-size $h$. 
In general, 
it makes sense to attempt
to obtain alternative (to the prior $\mu_0$) Gaussian reference 
measures that deliver `flattened' log-densities 
for the target $\mu$.
This leads naturally to the choice of \emph{local}
reference measures, as differently oriented elliptic 
contours can provide better proxies to the target contours at different parts of the state space.

We turn at this point to a finite-dimensional context 
(so $\Hi\equiv \mathbb{R}^{n}$ for some $n\ge 1$)
and adopt an informal approach to avoid 
distracting technicalities.
Assume that we are interested in 
the target posterior in the vicinity of
$u_0\in \Hi$.
A second-order Taylor expansion
of the log-target (up to an additive constant):
\begin{equation*}
l(u) := -\Phi(u) - \tfrac{1}{2}\langle u, \Co^{-1}u \rangle 
\end{equation*}
around $u_0$ will give that:
\begin{multline*}
\exp\{l(u)\}=
\\ =  c(u_0) 
 \exp\big\{   
- \tfrac{1}{2}\big\langle u-m(u_0),[-D^{2}l(u_0)](u-m(u_0)) 
\big\rangle
+\mathcal{O}(|u-u_0|^{3})
\,\big\}
\end{multline*}
for some easy-to-identify  $m(u_0)\in \Hi$, $c(u_0)\in\mathbb{R}$ that depend on $u_0$.
Thus, 
with the Gaussian law
$\mathcal{N}(m(u_0),[-D^{2}l(u_0)]^{-1})$ as new reference measure,
the negative 
log-density (w.r.t.\@ this Gaussian law) of the target $\mu$ 
will be equal to $c'(u_0) + \mathcal{O}(|u-u_0|^{3})$ for some constant $c'(u_0)\in\mathbb{R}$, 
i.e.,\@ relatively flat in the vicinity of $u_0$.
Following the discussion in the previous paragraph, 
we will aim to develop algorithms driven by these 
local reference measures.
(Note that this local Gaussian reference measure coincides  
with the local Gaussian approximation used in the development 
of the Stochastic Newton  method in \cite{martin12}.)

To be more specific, 
we will achieve the required effect by 
allowing for general location-specific preconditioner 
$K=K(u_0)$ with the choice of $K(u_0)^{-1}$ motivated by the structure of the
negative Hessian
$-D^{2}l(u_0)$ at current position 
 $u_0$.
Thus, we will work with the local reference 
measure (in the vicinity of $u_0$):
\begin{equation*}
\tilde{\mu}_0= \mathcal{N}(m(u_0),K(u_0))
\end{equation*}
($m(u_0)$ cancels out in the subsequent developments and will not affect the algorithms)
and the target distribution $\mu$ expressed as:
%
%
\begin{gather}
\label{eq:new}
\frac{d\mu}{d\tilde{\mu}_0}(u) = 
c''(u_0)\exp\{-\widetilde{\Phi}(u;u_0)\}\ ,
\end{gather}
for some $c''(u_0)\in\mathbb{R}$,
where we have defined the negative log-density:
\begin{align*}
 \widetilde{\Phi}(u;u_0):= \Phi(u) + \tfrac{1}{2}&\langle u,\Co^{-1}u \rangle  
- \tfrac{1}{2}\,\big\langle u-m(u_0),K(u_0)^{-1}(u-m(u_0)) 
\big\rangle \ ,
\end{align*}
indicating the discrepancy  between the target and the local reference measure.
%
We also write its derivative:
\begin{align}\label{eq:dphi_approx}
D\widetilde{\Phi}(u;u_0)= D\Phi(u) + 
\mathcal{C}^{-1}u - K(u_0)^{-1}(u-m(u_0)) \ . 
\end{align}
%
%
%
We will use the reference 
measures $\tilde{\mu}_0$ as drivers 
for the implicit scheme when deriving a local-move
MALA algorithm. 
Similarly to Section \ref{sec:MCMC},
we will also define an HMC-type algorithm
as an extension of the MALA version when we allow 
the synthesis of a number of local steps before applying 
the accept/reject.



\subsection{$\infty$-mMALA}
\label{sec:mMALA}
Recall the Langevin dynamics 
in (\ref{eq:Langevin}) that gave rise (for $K=\mathcal{C}$) to $\infty$-MALA in Section 
\ref{sec:MCMC}.
The above discussion, and re-expression 
of the target as in (\ref{eq:new}), 
suggest invoking dynamics of the type:
\begin{equation}\label{eq:Langevin2}
\frac{du}{dt} = -\frac12\,K(u)\,\big\{ \Co^{-1}u+ D\Phi(u)\big\} + \sqrt{K(u)}\, \frac{dW}{dt}
\end{equation}
for a location-specific preconditioner 
$K(u)$ (its choice motivated in practice by the form of the inverse negative Hessian at the current position).  
Notice that these dynamics do not, in general, preserve
the target $\mu$ as they omit the 
higher order (and computationally expensive) Christofell symbol terms, see e.g.\@
\cite{girolami11} and the discussion in 
\cite{xifara:14}.
As noted with the study of `Simplified MALA' in 
\cite{girolami11}, the dynamics in (\ref{eq:Langevin2})
can still capture an important part of the local curvature structure of the target and can provide an effective balance between mixing and computational cost. 
 
The time-discretization scheme develops 
as in the case of $\infty$-MALA, 
with the important difference that it will now be driven by the local reference measure $\tilde{\mu}_0$ 
rather than the prior.
%
That is, we re-write (\ref{eq:Langevin2})
as follows:
\begin{equation}\label{eq:Langevin3}
\frac{du}{dt} = -\frac12\,K(u)\,\big\{ 
K(u)^{-1}(u-m(u)) + D\widetilde{\Phi}(u;u)
\big\} + \sqrt{K(u)}\, \frac{dW}{dt}
\end{equation}
and develop the semi-implicit scheme as follows:
\begin{gather}
\label{eq:semi-implicit2}
\frac{u'-u}{h} = -\frac12\,\Big\{\frac{u+u'}{2} - m(u) + K(u) D\widetilde{\Phi}(u;u)\Big\} + \sqrt{\frac{1}{h}}\,\xi\  ;\\\ 
 \xi\sim \mathcal N(0,K(u))\ . \nonumber
\end{gather}
Notice that $m(u)$ cancels out
(simply apply operator $K(u)$ on both sides 
of \eqref{eq:dphi_approx}, replace $u_0\leftrightarrow u$ 
and use the obtained expression for  $K(u) D\widetilde{\Phi}(u;u)$
here)
and we can 
rewrite (\ref{eq:semi-implicit2}) 
in the general form:    

%
%
\begin{equation}
\frac{u'-u}{h} = -\frac{1}{2}\big\{\frac{u'+u}{2} - g(u)\big\} + \sqrt{\frac 1h}\,\xi\ ,\quad \xi\sim \mathcal N(0,K(u))\ ,
\end{equation}
where we have defined:
\begin{equation}\label{eq:ngrad}
g(u) = -K(u) \big\{(\Co^{-1}-K^{-1}(u))u + D\Phi(u)\big\}\  .
\end{equation}
%
%
Re-arranging terms, we can equivalently write:
\begin{equation}\label{eq:infmMALA}
u' = \rho\,u + \sqrt{1-\rho^2} \,v\ , \quad v=\xi +  \tfrac{\sqrt{h}}{2} g(u)\ ,\quad  \xi\sim \mathcal N(0,K(u))\ ,
\end{equation}
for $\rho$ defined as in (\ref{eq:infMALA}).


%
%
Recall the steps for identifying the Metropolis-Hastings acceptance 
probability in (\ref{eq:MH1})-(\ref{eq:MH3}) and the related notation 
for the involved bivariate measures.
The following assumptions are sufficient for the well-posedness of the 
proposal (\ref{eq:infmMALA}) and for providing a non-trivial 
Radon-Nikodym derivative 
$(d\tp{\nu}/d\nu)(u,u')$  on the Hilbert space $\Hi$.
\begin{asump}\label{as:1}
We have, $\mu_0$-a.s.~in  $u\in \Hi$,  that $K(u)$ is a 
self-adjoint, positive-definite and trace-class operator on 
Hilbert space $\Hi$, and it is such that:
\begin{itemize}
\item[i)] $\mathrm{Im}(K(u)^{1/2}) = \mathrm{Im}(\Co^{\frac12})(=\Hi^{\kappa})$;
\item[ii)] $\{\Co^{-1/2}K(u)^{1/2}\}\{  \Co^{-1/2}K(u)^{1/2} \}^{\top} - I$ is a 
Hilbert-Schmidt operator on $\Hi$.
\end{itemize}
%
\end{asump}
\noindent A linear, bounded operator $A:\Hi\mapsto \Hi$ is Hilbert-Schmidt if 
$\sum_j |A\phi_j|^2<\infty$.
\begin{asump}\label{as:2}
 $(K(u)\Co^{-1}-I)u\in \mathrm{Im}(\Co^{\frac12})(=
\Hi^{\kappa})$, $\mu_0$-a.s.~in $u$.
\end{asump}
\begin{cor}
\label{cor:1}
Under Assumptions \ref{ass:10}-\ref{as:2}, we have that $g(u)\in\Hi^{\kappa}$.
\end{cor}
\begin{proof}
Due to Assumption \ref{as:2}, it remains
to show $K(u)D\Phi(u)\in \Hi^{\kappa}$.
Note  that $K(u)D\Phi(u) =\Co^{1/2}R\,\Co^{1/2}D \Phi(u)$ where  
$R=\{\Co^{-1/2}K(u)^{1/2}\}\{  \Co^{-1/2}K(u)^{1/2} \}^{\top}$.
Also, from Assumption \ref{ass:10},  
$\Co^{1/2}D \Phi(u)\in \Hi^{\kappa-\ell}\subseteq \Hi$.
So, $\Co^{1/2}R\,\Co^{1/2}D \Phi(u)\in 
\mathrm{Im}(\Co^{1/2})= \Hi^{\kappa}$.
\end{proof}

From the Feldman-Hajek theorem (see e.g.\@ Theorem 2.23 in \cite{da14}),
Assumption \ref{as:1} and Corollary \ref{cor:1} are necessary and sufficient so that 
$\mathcal N(g(u),K(u))\simeq \mathcal N(0,\Co)$, $\mu_0$-a.s.~in~$u$.
The following result gives the corresponding 
Radon-Nikodym derivative,
which will then be used to illustrate the 
well-posedness of the MCMC algorithm  and 
provide the Metropolis-Hastings acceptance probability.
%
%
\begin{thm}
\label{th:mm}
Assumptions \ref{ass:10}-\ref{as:2} imply that 
$\mathcal N( (\sqrt{h}/2)\,g(u),K(u)) \simeq \mathcal N(0,\Co)$,
$\mu_0$-a.s.~in $u$, with Radon-Nikodym derivative:
\begin{align*}
\lambda(w;u): &= \frac{d\mathcal N( \tfrac{\sqrt{h}}{2}\,g(u),K(u))}{d\mathcal N(0,\Co)}
(w) = \frac{d\mathcal N(\tfrac{\sqrt{h}}{2}\,g(u),K(u))}{d\mathcal N(0,K(u))}(w)\times \frac{d\mathcal N(0,K(u))}{d\mathcal N(0,\Co)}(w)\\
&= 
\exp\big\{-
\tfrac{h}{8}|K^{-\frac12}(u)g(u)|^2+
\tfrac{\sqrt{h}}{2}\langle K^{-\frac12}(u)g(u),K^{-\frac12}(u)w\rangle\big\} \\
& \qquad \qquad \times 
 \exp\big\{-\tfrac12\langle w, (K^{-1}(u)-\Co^{-1})w \rangle \big\}\cdot 
|\,\Co^{1/2}K(u)^{-1/2}\,|\ . 
\end{align*}
\end{thm}
\begin{proof}
The first Radon-Nikodym derivative in the
expression for $\lambda(w,u)$ is an application of Theorem 2.21 of \cite{da14} on translations 
of Gaussian measures.
The second density is a formal expression of the ratio
of two Gaussian measures. 
\end{proof}
\begin{rk}
Note that due to the Hilbert-Schmidt property in Assumption \ref{as:1},
the term
\begin{equation}
\label{eq:dens1}
\langle w, (K^{-1}(u)-\Co^{-1})w \rangle -\log|\,\Co\, K(u)^{-1}|
\end{equation}
appearing in the expression for $\lambda(w,u)$ in Theorem \ref{th:mm} 
is a.s.\@ finite under $w\sim \mu_0$ ($\mu_0$-a.s.\@ in $u\sim \mu_0$) as expected 
(since we assume existence of a density).
For instance, the second moment of (\ref{eq:dens1}) is equal to 
(we use the standard representation on $\mathbb{R}^{n}$ 
by projecting onto the first 
$n$ basis functions in $\{\phi_i\}$;
we also denote by $\{\nu_{j,n}\}_{j=1}^{n}$ the eigenvalues of the 
projection  $\{\Co^{-1/2}K(u)^{1/2}\}\{  \Co^{-1/2}K(u)^{1/2} \}^{*}$ on $\mathbb{R}^{n\times n}$):
\begin{equation*}
a_n :=\Big\{\sum_{j=1}^{n}\big(\log \nu_{j,n} + \nu_{n,j}^{-1}-1  \big)
\Big\}^2
+ 2 \sum_{j=1}^{n} (\nu_{n,j}^{-1}-1)^2
\end{equation*}
From the Hilbert-Schmidt assumption we have that 
$\sup_{n}\sum_{j=1}^{n}(1-\nu_{j,n})^2<\infty$, thus also 
$C_1 \le \sum_{j,n}\nu_{j,n}\le C_2$, for constants $C_1,C_2>0$.
Since $0\le (\log \nu_{j,n} + \nu_{n,j}^{-1}-1) \le C\,(1-\nu_{j,n})^2$
for some constant $C>0$, we have that $\sup_{n}a_n <\infty$.
\end{rk}

Let $Q(u,du')$ being the proposal kernel 
derived from \eqref{eq:infmMALA}; we also consider the bivariate measure  
$\nu(du,du')=\mu(du)Q(u,du')$.
Recall from (\ref{eq:MH1})-(\ref{eq:MH3})
that obtaining the Metropolis-Hastings 
accept/reject rule requires 
finding the Radon-Nikodym derivative 
$d\nu^{\top}/d\nu$.
Similarly to the derivation of $\infty$-MALA in Section~\ref{sec:MCMC} we consider now the bivariate Gaussian law
${\nu}_0(du,du')=\mu_0(du)Q_0(u,du')$
with $Q_0(u,du')$ as in (\ref{eq:refG}).
Recall  we have the symmetry property $\nu_0\equiv\nu_0^{\top}$.
Applying Theorem \ref{th:mm} we have:
\begin{equation}
\label{eq:denG}
\frac{d\nu}{d\nu_0}(u,u') = 
\frac{d\mu}{d\mu_0}(u)\cdot 
\frac{dQ(u,\cdot)}{dQ_0(u,\cdot)}(u') =
\frac{1}{Z}\exp\{-\Phi(u)\}\cdot  
\lambda(\tfrac{u'-\rho u}{\sqrt{1-\rho^2}};u)\ .
\end{equation}
%
We obtain 
the required density 
as $(d\nu^{\top}/d\nu) = 
[\,d\nu^{\top}/d{\nu}_0^{\top}]\,/\,[\,d\nu/d{\nu_0}\,]$.
We can now define the complete method, labeled $\infty$-mMALA 
in Algorithm \ref{tab:mmMALA}, (the small `m' in the name stands for `manifold').

\begin{algorithm}[!h]
\begin{flushleft}
\medskip
{\itshape
\begin{enumerate}
\item[\textit{1.}] Given current $u$, sample independently  
$\xi\sim \mathcal N(0,K(u))$
and propose: 
\begin{equation*}
u' = \rho u + \sqrt{1-\rho^2} \big\{ \xi +  \tfrac{\sqrt{h}}{2} g(u) \big\}\ . 
\end{equation*} 
\item[\textit{2.}]  Accept $u'$ with probability  $a(u,u') = 1\wedge \frac{\kappa(u',u)}{\kappa(u,u')}$, where we have set:
\begin{align*}
\kappa(u,u') = \frac{1}{Z}\,\exp\{-\Phi(u)\}\times
\lambda(\tfrac{u'-\rho u}{\sqrt{1-\rho^2}};u) 
\end{align*}
otherwise stay at $u$.
%
%
\end{enumerate} }
\medskip
\end{flushleft}
\vspace{-0.4cm}
\caption{A single Markov step for $\infty$-mMALA.}
\label{tab:mmMALA}
\end{algorithm}
\begin{rk}
When $K(u)\equiv \Co$, algorithms $\infty$-MALA and 
$\infty$-mMALA coincide. 
\end{rk}
%
%
%
%
In the following we let ${\mathsf H}(u)$ denote the posterior Hessian, computed
from the negative log posterior:
$${\mathsf H}(u):=\Co^{-1}+ D^2\Phi(u)\ ;$$
since this is not necessarily positive-definite it is also of interest
to consider a modification in which the non-positive and small eigenvalues are
all shifted above a threshold, as in \cite{martin12}, and we use
the same notation ${\mathsf H}(u)$ for this modification in order not
to clutter notation. The following corollary connects our methodology
with the Stochastic Newton (SN) MCMC method from \cite{martin12}.
We also recall that the paper \cite{petra14} considered variants
on this method where ${\mathsf H}(\cdot)$ is evaluated at the MAP
point, and low rank approximations are employed.

\begin{cor}\label{cor:equiv2SN}
When $\rho=0$ ($h=4$), $\infty$-mMALA coincides with the SN MCMC method.
\end{cor}
\begin{proof}
When $\rho=0$, we have $h=4$ from \eqref{eq:infMALA}. 
The proposal \eqref{eq:infmMALA} of $\infty$-mMALA becomes:
\begin{gather}
u' \sim \mathcal N(g(u),K(u))\ , \quad g(u) = u -K(u) ( \Co^{-1} u + D\Phi(u) )\ , \nonumber \\
\quad  K(u) = {\mathsf H}(u)^{-1} \label{eq:equiv2SN}
\end{gather}
which is exactly the proposal for the SN MCMC method defined in Section 2.3
of \cite{martin12}.
\end{proof}

\subsection{$\infty$-mHMC}
\label{sec:inf-mHMC}

Following the same direction as with $\infty$-mMALA, 
we now begin from the continuous-time Hamiltonian dynamics in (\ref{Hamiltonian}), with a location-specific mass matrix:
\begin{equation}
\label{mHamiltonian0}
\frac{d^2u}{dt^2} + K(u)\,
\big\{\, \mathcal{C}^{-1}u + D\Phi(u) \big\} = 0\ .
\end{equation}
For a splitting scheme driven by the local Gaussian 
reference measure $\tilde{\mu}_0$, 
we re-write the above dynamics as:
\begin{equation}
\label{mHamiltonian1}
\frac{d^2u}{dt^2} + K(u)\,
\big\{\, K(u)^{-1}(u-m(u)) + D\widetilde{\Phi}(u;u) \big\} = 0\ .
\end{equation}
As with $\infty$-mMALA, $m(u)$ cancels out.
Setting $du/dt=v$, we make use of the following 
splitting scheme:
\begin{align}
\label{eq:msplit}
du/dt = v \ , \quad 
& dv/dt = -u \ ; \\
\label{eq:msplit1}
du/dt = 0 \ , \quad 
& dv/dt = -K(u)\big\{\,(\mathcal{C}^{-1}-K^{-1}(u))\,u+ D\Phi(u) 
\,\big\} \ . 
\end{align}
Both (\ref{eq:msplit}), (\ref{eq:msplit1}) 
can be solved analytically, the first by applying a rotation. 
Thus, we obtain the following approximate symplectic integrator of (\ref{mHamiltonian0}),
%
%
for $g$ as defined in \eqref{eq:ngrad}:

\begin{equation}\label{mHDdiscret}
\begin{aligned}
v^- &= v_0 + \tfrac{\epsilon}{2}\,g(u_0)\ ; \\
\begin{bmatrix} u_\epsilon\\ v^{+}\end{bmatrix} &= \begin{bmatrix} \cos\epsilon & \sin\epsilon\\ -\sin\epsilon & \cos\epsilon
\end{bmatrix}  \begin{bmatrix} u_0\\ v^{-}\end{bmatrix}\  ;\\
v_\epsilon &= v^{+} + \tfrac{\epsilon}{2}\,g(u_\epsilon)\  .
\end{aligned}
\end{equation}
%
Equation \eqref{mHDdiscret} gives rise to the
leapfrog map $\Psi_\epsilon: (u_{0},v_{0})\mapsto (u_{\epsilon},v_{\epsilon})$.
Given a time horizon $\tau$ and current position 
$u$, the MCMC mechanism proceeds 
by proposing:
\begin{equation*}
u' =\mathcal{P}_u\big\{{\Psi}_{\epsilon}^{I}(u,v)\big\}\ , \quad v\sim\mathcal{N}(0,K(u))
\ . 
\end{equation*}
for $I=\lfloor \tau/\epsilon \rfloor$.
%
%
%
%
Note that the dynamics in (\ref{mHamiltonian0}) 
do not preserve, in general, the target distribution~$\mu$ 
(when initialized with $v\sim \mathcal{N}(0,K(u))$). 
Thus, there is no theoretical guarantee 
that the algorithm will give good 
acceptance probabilities for arbitrary time lengths
$\tau$ with diminishing $\epsilon$ - an important 
property that characterises non-local HMC algorithms.
However, with properly chosen $\tau$,
$\infty$-mHMC, as a multi-step generalization of $\infty$-mMALA 
(see the similar discussion in Section \ref{sec:MCMC} and the formal statement in Remark \ref{cor:multi-step} below), 
is a valuable algorithm to be tested in applications,
and in the numerical examples that follow it is 
indeed found in many cases to be superior than $\infty$-mMALA. 
%

The following theorem is required for  establishing the well-posedness of the developed algorithm. We define the probability measures on the phase-space:
\begin{align*}
S_0(du,dv) &:= \mathcal{N}(0,\Co)(du)\otimes \mathcal{N}(0,\Co)(dv)\ ; \\[0.1cm]
\tilde{S}_0(du,dv) &:= \mathcal{N}(0,\Co)(du)\otimes
 \mathcal{N}(\tfrac{\epsilon}{2}g(u),\Co)(dv)\ ; \\[0.1cm]
S(du,dv) &:= \mu(du) \otimes
  \mathcal{N}(0,K(u))(dv) \ .
\end{align*}
We also define the push-forward probability measures:
\begin{align*}
S^{(i)}: = S\circ \Psi_\epsilon^{-i}\ ,
\quad 1\le i\le I\ .  
\end{align*}
For starting positions $u_0, v_0$, we set $(u_i,v_i):=\Psi_\epsilon^{i}(u_0, v_0)$, $0\le i\le I$.

\begin{thm}\label{thm:infmHMC}
\begin{itemize}
\item[(i)]
Under Assumptions \ref{ass:10}-\ref{as:2}, 
Theorem \ref{th:mm} implies the absolute continuity
$S^{(i)}\simeq S_0$, for all indices $1\le i\le I$, with Radon-Nikodym derivatives satisfying the recursion:
\begin{equation*}
\frac{dS^{(i)}}{dS_0}(u_i,v_i)=
 \frac{dS^{(i-1)}}{dS_0}(u_{i-1},v_{i-1})
 \cdot 
 G(u_{i-1},v_{i-1}+\tfrac{\epsilon}{2}g(u_{i-1}))
 \cdot G(x_i, v_i)
\end{equation*}
where we have defined:
\begin{equation*}
G(u,v) = \frac{d\tilde{S}_0}{dS_0}(u,v) =
\exp\big\{-
\tfrac{\epsilon^2}{8}|\Co^{-\frac12}g(u)|^2+
\tfrac{\epsilon}{2}\langle \Co^{-\frac12}g(u),\Co^{-\frac12}v\rangle\big\} \ . 
\end{equation*}
\item[(ii)]
From (i) we obtain that:
\begin{align*}
\frac{dS^{(I)}}{dS}(u_I,v_I)= 
\frac{(dS/dS_0)(u_0,v_0)}{(dS/dS_0)(u_I,v_I)}\times
\prod_{i=1}^{I} G(u_{i-1},v_{i-1}+\tfrac{\epsilon}{2}g(u_{i-1}))
 \cdot G(u_i, v_i)\ . 
\end{align*}
%
%
We can re-write:
\begin{align*}
\log\big\{ (dS^{(I)}/dS)(u_I,v_I) \big\} =\Delta H(u_0,v_0)
\end{align*}

for the following quantity:
\begin{align*}
\Delta H(&u_0,v_0) 
=\Phi(u_I)-\Phi(u_0)+\tfrac{1}{2}
\langle v_I, (K^{-1}(u_I)-\Co^{-1})v_I\rangle -
\tfrac{1}{2}\langle v_0, (K^{-1}(u_0)-\Co^{-1})v_0\rangle \\[0.2cm]
&-\log|\Co^{1/2}K^{-1/2}(u_I)|
+\log|\Co^{1/2}K^{-1/2}(u_0)|
-\tfrac{\epsilon^2}{8}\big(\,|\Co^{-\frac12}g(u_I)|^2-
|\Co^{-\frac12}g(u_0)|^2\,\big) \\
&+\tfrac{\epsilon}{2} \sum_{i=0}^{I-1}\big(\,\langle \Co^{-1/2}g(u_{i}), \Co^{-1/2}v_i\rangle 
+ \langle \Co^{-1/2}g(u_{i+1}),\Co^{-1/2}v_{i+1}%
\rangle\,\big)  \ . 
\end{align*}
\item[(iii)]
We have the identity:
\begin{align*}
\Delta H(u_0,v_0) \equiv H(u_I, v_I) - H(u_0,v_0)
\end{align*}
for the energy function:
\begin{equation*}
H(u,v):= \Phi(u) + \tfrac{1}{2}\langle u, \Co^{-1}u \rangle + \tfrac{1}{2}\langle v, K(u)^{-1}v \rangle 
- \log |\Co^{1/2}K(u)^{-1/2}|\ . 
\end{equation*}
\item[(iv)] 
Given current position $u\in\Hi$, the Markov chain with proposed move:
\begin{equation*}
u' =\mathcal{P}_u\big\{{\Psi}_{\epsilon}^{I}(u,v)\big\}\ , \quad v\sim\mathcal{N}(0,K(u))
\ ,
\end{equation*}
and acceptance probability:
\begin{equation*}
a  = 1\wedge \exp\{-\Delta H(u,v)\}
\end{equation*}
preserves the target probability measure $\mu$.
\end{itemize}
\end{thm}
\begin{proof}
See \ref{appdx:thm-infmHMC}.
\end{proof}

We can now define the complete method, labeled $\infty$-mHMC, 
in Algorithm \ref{tab:mHMC} below.

\begin{algorithm}[h]
\begin{flushleft}
\medskip
{\itshape
\begin{enumerate}
\item[\textit{1.}] Given current $u$, sample independently  $v\sim \mathcal N(0,K(u))$
and propose 
$u' = \mathcal{P}_u\big\{ \Psi_\epsilon^{I}( u,v )\big\}$. 
\item[\textit{2.}]  Accept $u'$ with probability $1\wedge \exp\big\{-\Delta H(u,v)\}$
for the change of energy $\Delta H(u,v)$ defined in Theorem 
\ref{thm:infmHMC} (ii)-(iii), 
%
otherwise stay at $u$.
%
%
\end{enumerate} }
\medskip
\end{flushleft}
\vspace{-0.4cm}
\caption{A single Markov step for $\infty$-mHMC.}
\label{tab:mHMC}
\end{algorithm}

\begin{rk}
When $K(u)\equiv \Co$, algorithms $\infty$-HMC and 
$\infty$-mHMC coincide. 
\end{rk}

\begin{cor}\label{cor:multi-step}
Assume that we allow for different step-sizes in the  leapfrog scheme in (\ref{mHDdiscret}):  $\epsilon_1$ in the first and third equation, and $\epsilon_2$ in the second (the rotation). Recall 
the step-size $h$ in the definition of $\infty$-mMALA.
Then, if $I=1$, and $\epsilon_1$ and $\epsilon_2$ are such that:
\begin{equation}\label{step-size-setting}
\epsilon_1^2=h\ , \quad
\cos\epsilon_2 = \frac{1-\epsilon_1^2/4}{1+\epsilon_1^2/4}\ , \quad \sin\epsilon_2 = \frac{\epsilon_1}{1+\epsilon_1^2/4}\ ,
\end{equation}
 algorithms $\infty$-mMALA and $\infty$-mHMC coincide.
\end{cor}
\proof{See  \ref{appdx:cor-multi-step}}.

\begin{rk}
Following Corollary \ref{cor:equiv2SN} and Corollary \ref{cor:multi-step}, the following plot illustrates graphically the connections between the various algorithms.
\begin{align*}
\boxed{\infty\textrm{-MALA}} &\xrightarrow{\textrm{position-dependent preconditioner}\; K(u)} \boxed{\infty\textrm{-mMALA}} &\xrightarrow{h=4} \boxed{\textrm{SN}}\\
\rotatebox[origin=c]{270}{$\xrightarrow{\textrm{multiple steps}\; (I>1)}$} &\phantom{\xrightarrow{\textrm{position-dependent preconditioner}\; (K(u))}\quad} \rotatebox[origin=c]{270}{$\xrightarrow{\textrm{multiple steps}\; (I>1)}$} &\\
\boxed{\infty\textrm{-HMC}} &\xrightarrow{\textrm{position-dependent preconditioner}\; K(u)} \boxed{\infty\textrm{-mHMC}} &\\
\end{align*}
\end{rk}


\subsection{Split $\infty$-mMALA and $\infty$-mHMC}
Following the discussion on optimal local Gaussian approximation in Subsection \ref{sec:mMALA}
or the metric tensor interpretation in \cite{girolami11}, 
a typical choice of $K(u)^{-1}$ is the expectation over the data $y$ given $u$ 
of the negative Hessian of the log-target (this choice also guarantees positive-definiteness of $K(u)$), that is:
\begin{equation}\label{metric}
K(u)^{-1} = F(u) + \Co^{-1}, \quad F(u) := \E_{y|u}[\,D\Phi(u;y) \otimes D\Phi(u;y)\,]\  .
\end{equation}
%
Assuming a projection onto finite dimension $n\ge 1$, the operations of 
obtaining the operator $K(u)^{-1}$, applying it on a vector, inverting it or sampling from $\mathcal{N}(0,K(u))$ will typically have computational costs of order $\mathcal{O}(n^3)$ for each given
current  $u\in \Hi$. This can be prohibitively expensive when $n$ is large,
and could cause algorithms to be less efficient than simpler ones
that use a constant mass matrix, when compared according to cost
per independent sample.
However, in a large class of inverse problem applications, the 
typical wave-length of the eigenfunctions of the covariance $\Co$
decays as the eigenvalues decay (consider for example the periodic
setting where $\Co$ is an inverse fractional power of the Laplacian operator $\Delta$).  
As a consequence, for typical observations which inform low frequencies,
the information from the data spreads non-uniformly with respect to
the coordinates $\{u_i\}$ of the unknown function parameter $u$,
with most of it concentrating on the low-frequency coordinates.
We will take advantage of this setting to recommend an effective choice of preconditioner $K(u)$.
 
Recall the orthonormal basis $\{\phi_j\}$ of $\Hi$ consisting of 
eigenfunctions of $\Co$, and the isomorphism $\Hi \leftrightarrow \ell^2$
mapping $u\leftrightarrow \{u_j\}$ with $u= \sum_{j\ge 1}u_j\phi_j = 
\sum_{j\ge 1}\langle u, \phi_j \rangle \phi_j $.
%
%
%
For a cut-off point $D_0\ge 1$, 
we write $u=(u^t, u^r)$ with $u^t := u_{1:D_0}$ 
and  residual part $u^r:=u_{(D_0+1):\infty}$. 
%
We define the truncation operator $T$ mapping
\begin{equation}
\label{truncation}
 \quad u\mapsto (u^t,0,0,\ldots)\ 
\end{equation}
with domain $\Hi^{-\ell}$.
Balancing computational considerations with mixing effectiveness of the proposal move within 
the Metropolis-Hastings framework, we recommend using 
 the following truncated Fisher information operator:
\begin{equation}
\tilde{F}(u)  = \mathbb{E}_{y|u}[\,T\{D\Phi(u;y)\} \otimes T\{D\Phi(u;y)\}\,]\  .
 \end{equation}
Thus, we recommend the following choice:
\begin{equation}
\label{eq:Gnew}
K^{-1}(u) := \tilde{F}(u) + \Co^{-1}\  .
\end{equation}
Given that $\{\phi_j\}$ corresponds to the eigenfunctions of $\Co$,
operator  $K(u)$ in (\ref{eq:Gnew}) trivially satisfies 
Assumptions \ref{as:1}-\ref{as:2}, as $\tilde{F}(u)$ only has  
a finite-size  upper diagonal block of non-zero entries. Indeed, we can write:
\begin{equation*}
K(u)^{-1} = \left(  
\begin{array}{cc} 
\{K(u)^{t}\}^{-1} & 0 \\
0 & \{K(u)^{r}\}^{-1}
\end{array} 
\right) =  \left(\begin{array}{cc} 
\tilde{F}(u)^{t} + \Co^{t} & 0 \\
0 & \Co^{r}
\end{array} 
\right)\  .
\end{equation*}
with the truncations on the operators defined in the obvious way.
%

We label as \emph{Split $\infty$-mMALA} and \emph{Split $\infty$-mHMC}
the correponding MCMC methods resulting from the above choice of location specific preconditioner.
The calculation of all required algorithmic quantities is now simplified, 
due to $K(u)$ being diagonal except for a finite-range of values.
Indeed, in the case for instance of  Split $\infty$-mMALA, the proposal may be written as:
\begin{equation*}
(u^t, u^r)' = \rho\,(u^t, u^r) + \sqrt{1-\rho^2}\,\big\{ (\xi^{t}, \xi^{r}) + \tfrac{\sqrt{h}}{2}(g(u)^{t}, g(u)^{r}) \big\}
\end{equation*}
where we have:
\begin{gather*}
\xi^t\sim \mathcal{N}\big(0,K(u)^t \big) \ , \quad  \xi^r \sim \mathcal{N}(0,\Co^{r})\ ,\\[0.2cm]
g(u)^{t} = - K(u)^{t}\,\big\{  - \tilde{F}(u)^{t}\,u^{t} + D\Phi(u)^{t}  \big\}\ , \quad 
g(u)^{r} = - \Co^{r} D\Phi(u)^{r} \ . 
\end{gather*}

\begin{rk}
Splitting the proposal into a likelihood-informed and a simpler step bears similarities 
with the `intrinsic subspace' method in \cite{cui16}.
We stress however that our methodology develops geometric algorithms, 
in the sense that it employs location-specific curvature information.
The development of the geometric methods in a general setting in the earlier sections (beyond the truncation we recommend here) is still necessary for mathematical rigorousness, and more importantly, for the numerical robustness to possibly high-dimensional `intrinsic subspaces'.
As previously discussed, the straightforward splitting implemented here works fairly well on a class of inverse problems we consider in Section~\ref{sec:numerics}. 
We should mention here that in a context where the data in the inverse problem possess such strong information that a faithful representation of $u$ would require a large set  of high-frequency coordinates,
then more sophisticated likelihood-informed splitting methods, e.g. \cite{cui14}, \cite{constantine15a}, could potentially be considered to help derive low-dimensional `intrinsic subspaces'.
\end{rk}


\section{Numerical Experiments}
\label{sec:numerics}
Our experiments involve simulation studies based on three
physical inverse problems.
The prior is in each case Gaussian on a Hilbert space $\Hi$. 
In this section we consider three inverse problems -- the groundwater flow, the thermal fin heat conductivity and the laminar jet.
The first two examples are implemented in MATLAB (r2015b) and the last one is implemented in FEniCS \citep{alnaes15,logg12}.
All computer codes are available at \href{https://bitbucket.org/lanzithinking/geom-infmcmc}{https://bitbucket.org/lanzithinking/geom-infmcmc}.
The necessary adjoint and tangent linearized versions of this solver are derived with the dolfin-adjoint package \citep{farrell13}.

\subsection{Prior Specification}
\label{sec:gauss}
We will consider Hilbert spaces $\Hi\subseteq L^2(\mathcal D;\mathbb R)$, the latter denoting the space of real-valued squared-integrable functions on bounded open domains $\mathcal D\subset \mathbb R^d$, $d\ge 1$.
We denote by
$\langle\cdot,\cdot\rangle$ and $\Vert\cdot\Vert$ 
the inner product and norm, respectively, of $L^2(\mathcal D;\mathbb R)$.
In all of our examples we will construct the Karhunen-Lo\`eve (K-L) expansion \eqref{eq:KL}
through eigenfunctions of the Laplacian.
Specifically, we choose
covariance operators on $\Hi$ of the form: 
%
\begin{equation}\label{cov-op}
\sigma^2(\alpha\, \mathrm{I}-\Delta)^{-s}
\end{equation}
for scale parameters $\alpha,\sigma^2>0$, `smoothness' parameter $s\in\mathbb{R}$ 
and the Laplacian $\Delta=\sum_{j=1}^{d}
\partial_j^2 $.

In the first two numerical examples we have $d=2$,
and rectangular domain~$\mathcal{D}$ 
of the form $[k_1,k_2]\times[l_1,l_2]$
for integers $k_1, k_2, l_1, l_2$.
In this case, we will work with the 
orthonormal basis:
\begin{equation}
\phi_i({\bf x}) =2 |\mathcal{D}|^{-1/2}\,\cos\big\{\pi \big(i_1+\tfrac{1}{2}\big) x_1\big\}
\cos\big\{\pi \big(i_2+\tfrac{1}{2}\big) x_2\big\} \ ,
\quad i_1\ge 0\ ,\,\,i_2\ge 0\  .
\label{eigenbasis}
\end{equation}
Thus, the Hilbert space will 
be (we set $I=\{i=(i_1,i_2):i_1\ge 0, i_2\ge 0\}$): 
$$\Hi = \mathrm{span}\big\{\phi_i; i\in I\big\}\equiv 
\big\{u\in L^2(\mathcal{D};\mathbb{R}):
u=\sum_{i\in I}u_i \phi_i\ ,\, \sum_{i\in I}u_i^2<\infty  \big\}\ . $$
Guided by (\ref{cov-op}), we set the covariance operator $\Co$ as:
\begin{equation}
\label{eq:C}
\Co = \sum_{i\in I}\lambda_i^2 \{\phi_i\otimes \phi_i\}
\ ;\quad 
 \lambda_i^2 = \sigma^2\,\big\{ \alpha+\pi^2\big(\big(i_1+\tfrac{1}{2})^2 + 
 \big(i_2+\tfrac{1}{2}\big)^2\big)\big\}^{-s}\ .
\end{equation} 
For $\Co$ to be trace-class we require that 
$\sum_{i\in I}\lambda_i^2<\infty$, that is $s>1$.

In the third example we will have $d=1$, $\mathcal{D}=[-1,1]$
and use a prior covariance with the following orthonormal eigenfunctions and eigenvalues:
%
\begin{equation}
\phi_i(x) = \big(\tfrac{1}{\sqrt{2}}\big)^{\delta\,[\,i=0\,]}\cos(\pi i x)\ ,\quad 
\lambda_i^2 = 
2^{\,\,\delta\,[\,i=0\,]}\,
\sigma^2\,\{\,\alpha+(\pi i)^2\,\}^{-s}\ , \quad i\ge 0\ , 
\end{equation}
where $\delta\,[\cdot]$ is the indicator of whether condition(s) in the square bracket being satisfied (1), or otherwise (0), and
the trace-class property requires that $s>1/2$.


For given orthonormal basis $\{\phi_i\,;\, i\in I\}$, we run MCMC algorithms to sample K-L coordinates $\{u_i:=\langle u,\phi_i\rangle\,;\, i\in I_0\subset I\}$ in the following experiments, for some chosen non-negative integer $|I_0|$.
These coordinates can be viewed as projections of parameter function $u$ onto K-L modes up to $|I_0|$.
The splitting methods are implemented with Fisher operator truncated on the first $D_0$ of $|I_0|$ coordinates.
The gradient $D\Phi(u)$ is obtained by one adjoint solver in addition to the forward solution to the relevant PDE;
the metric action $\tilde{F}(u)\,v$ is obtained by another two extra adjoint (incremental) solvers for each $v\in \mathbb X$~\citep{bui14}.

\begin{figure}[t]
  \begin{center}
    \includegraphics[width=1\textwidth]{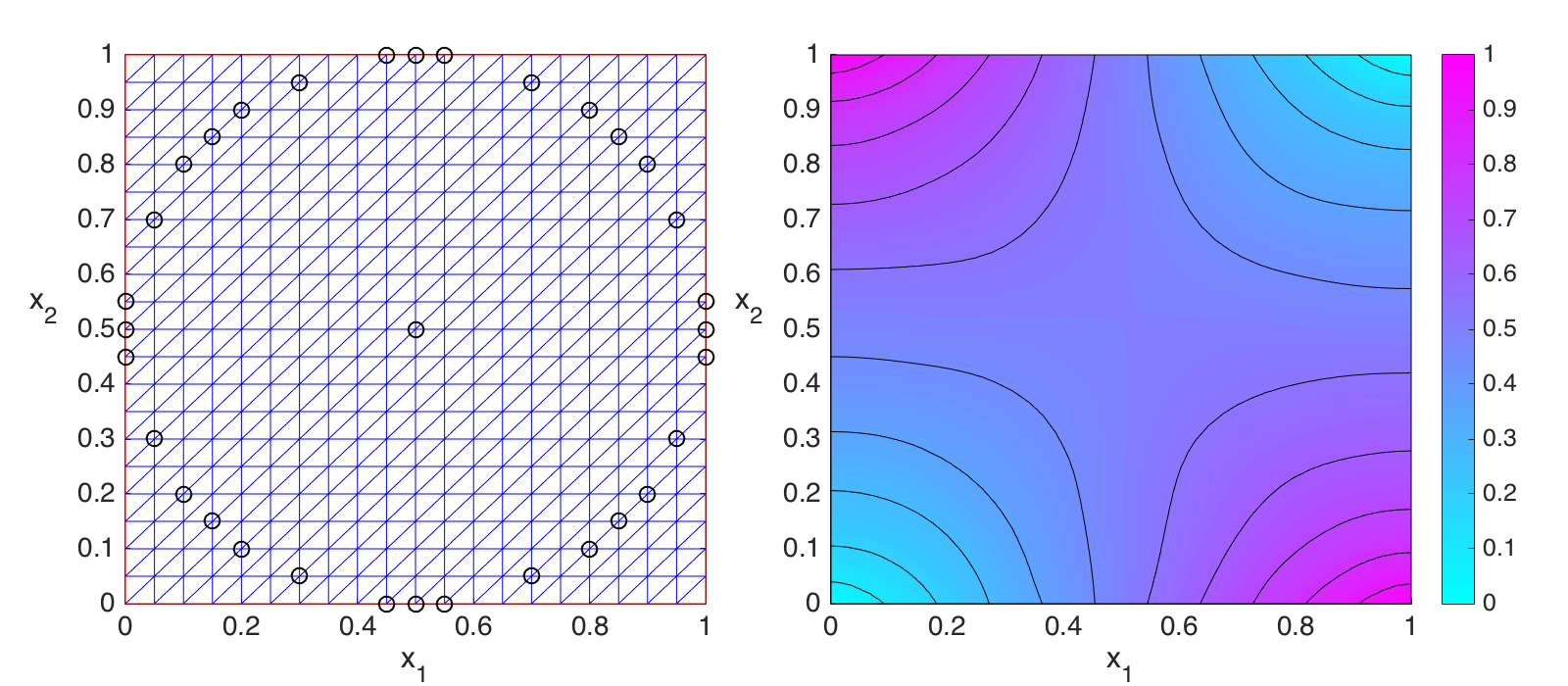}
  \end{center}
  \vspace{-0.5cm}
  \caption{Groundwater flow problem: the location of noisy observations (circles) on $[0,1]\times[0.1]$ (left) and the forward PDE solution under the true 
  permeability $u^{\dagger}$ (right).}
  \label{fig:obs_soln_groundwater}
\end{figure}

\subsection{Groundwater Flow}
We consider a canonical inverse problem involving the following elliptic PDE \citep{dashti11,conrad14} defined on the unit
square $\mathcal{D}=[0,1]^2$:
\begin{equation}
\label{groundwater}
\begin{aligned}
- \nabla\cdot (e^{u({\bf x})}\nabla p({\bf x})) &= 
0\  ;\\ p({\bf x})|_{x_2=0} & = x_1\ ; \\
p({\bf x})|_{x_2=1} & = 1-x_1\ ; \\
\left.\frac{\pa p({\bf x})}{\pa x_1}\right|_{x_1=0} &= \left.\frac{\pa p({\bf x})}{\pa x_1}\right|_{x_1=1} = 0\ . 
\end{aligned}
\end{equation}
This PDE serves as a simple model of steady-state flow in aquifers and other
subsurface systems. 
The unknown parameter $u$ represents the logarithm of permeability of the porous medium
and
$p$ represents the hydraulic head function.
The inverse problem involves inferring the log-permeability field $u=u(\mathbf{x})$ based on noisy observations,
$y$, of $p=p(\mathbf{x})$.

We consider a Gaussian prior on  $\Hi \subset L^{2}( \mathcal{D};\mathbb{R})$
with covariance $\Co$ of eigen-structure 
$(\lambda_j^2, \phi_j)_{j\in I}$, 
as explained in Subsection \ref{sec:gauss}.
We pick hyper-parameter values $\alpha=0$, $s=1.1$, $\sigma^2=1$.
To generate the data, we choose 
the true log-permeability field $u^{\dagger}$
via its coordinates
 $u_i^{\dagger}= \lambda_i^{1/2} \sin\big((i_1-1/2)^2+(i_2-1/2)^2\big)\cdot 
\delta\,[\,1\le i_1,i_2\le 10\,]$. 
In this setting, we solve the forward equation \eqref{groundwater} on a $40\times 40$ mesh and add Gaussian noise to 33 positions, $\mathbf{x}_n$, $1\le n \le 33$, of the true 
hydraulic head function $p^{\dagger}$ located on a circle and shown
on the left panel of Figure \ref{fig:obs_soln_groundwater}. In particular, 
we simulate data as:
\begin{equation*}
y_n = p^{\dagger}({\bf x}_n) + \eps_n\ , \quad \eps_n\sim \mathcal N(0,\sigma_y^2)\ ,
\end{equation*}
with $\sigma_y^2=10^{-4}$.
When running the MCMC algorithms,
the posterior is approximated by projecting 
the coordinates on $I_0=\{i\in I:i_1\le 10,\,i_2\le 10\}$  
and applying the PDE solver on a $20\times 20$ mesh.

We run the MCMC algorithms: pCN, $\infty$-MALA, $\infty$-HMC, $\infty$-mMALA, $\infty$-mHMC, Split $\infty$-mMALA and Split $\infty$-mHMC. For the split methods we truncate at $i_1, i_2\le 5$ based on threshing the eigenvalues $\{\lambda_i^2\}$ of $\Co$. Therefore we have $|I_0|=100$ and $D_0=25$ for this example.
Each algorithm is run for $1.1\times10^4$ iterations,
with the first $10^3$-samples used for burn-in.
HMC algorithms use a number of leapfrog steps 
chosen at random between 1 and 4.
All steps-sizes were tuned to obtain acceptance probabilities 
of about $60\%$-$70\%$.

\begin{figure}
  \centering
  \includegraphics[width=1\textwidth]{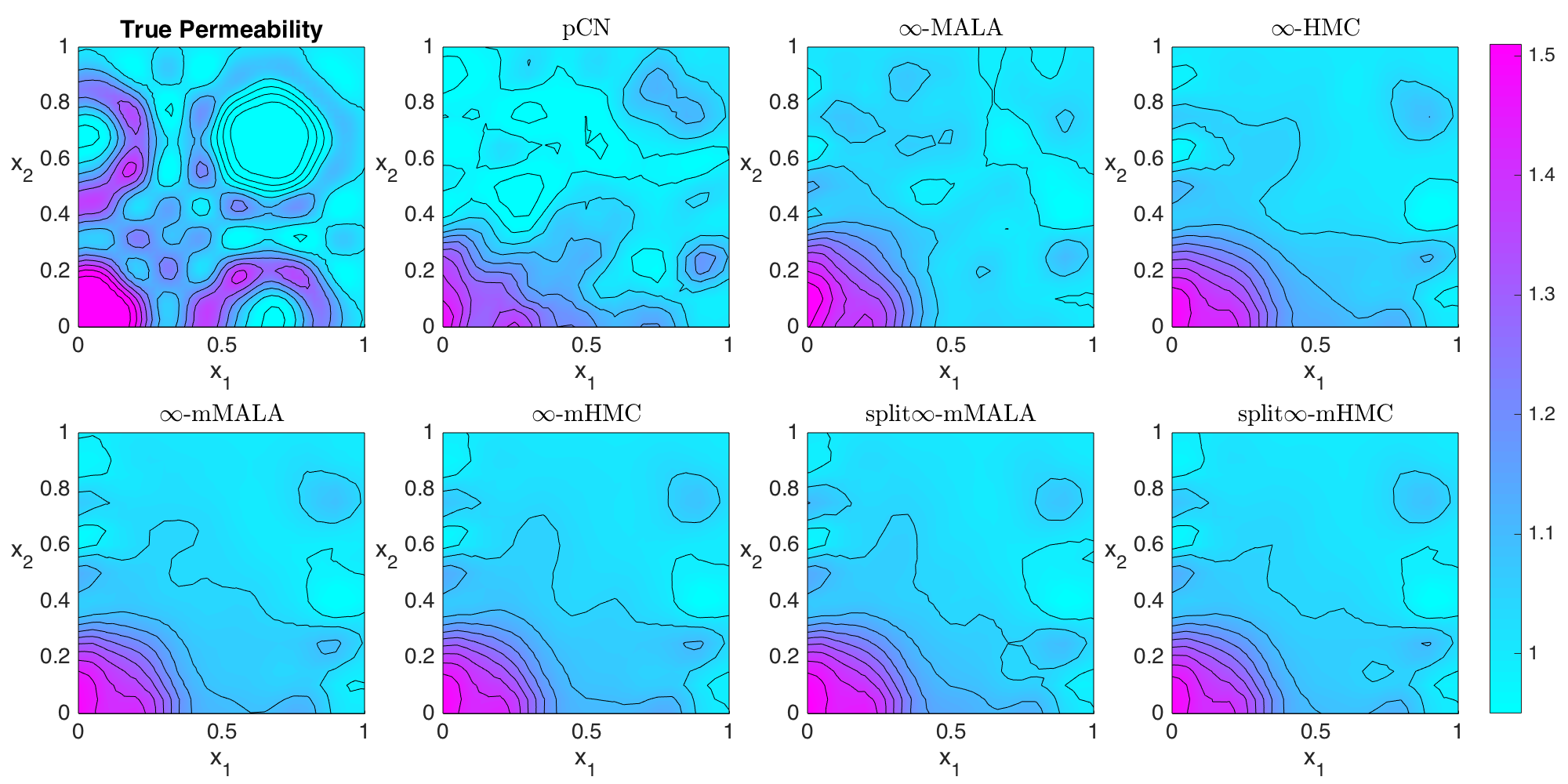}
\vspace{-1cm} 
  \caption{Groundwater flow problem: the true permeability field $e^{u^{\dagger}}$ (upper left-most) 
  and the posterior mean estimates provided by the various MCMC algorithms.}
  \label{fig:truth_est_groundwater}
\end{figure}
Figure \ref{fig:truth_est_groundwater} illustrates the posterior mean estimates of the permeability of the porous medium provided by the various algorithms.
The estimates by pCN and $\infty$-MALA differ from the rest due to the bad convergence properties of these algorithms. Figure \ref{fig:misfit_acf_groundwater}
shows the traceplots and corresponding autocorrelation functions for the negative log-likelihood $\Phi(u)$ (or `data-misfit') evaluated at the sample values;
the various traces are vertically offset to allow for comparisons. 
%
%
%
\begin{figure}[!h]
  \centering
  \includegraphics[width=1\textwidth]{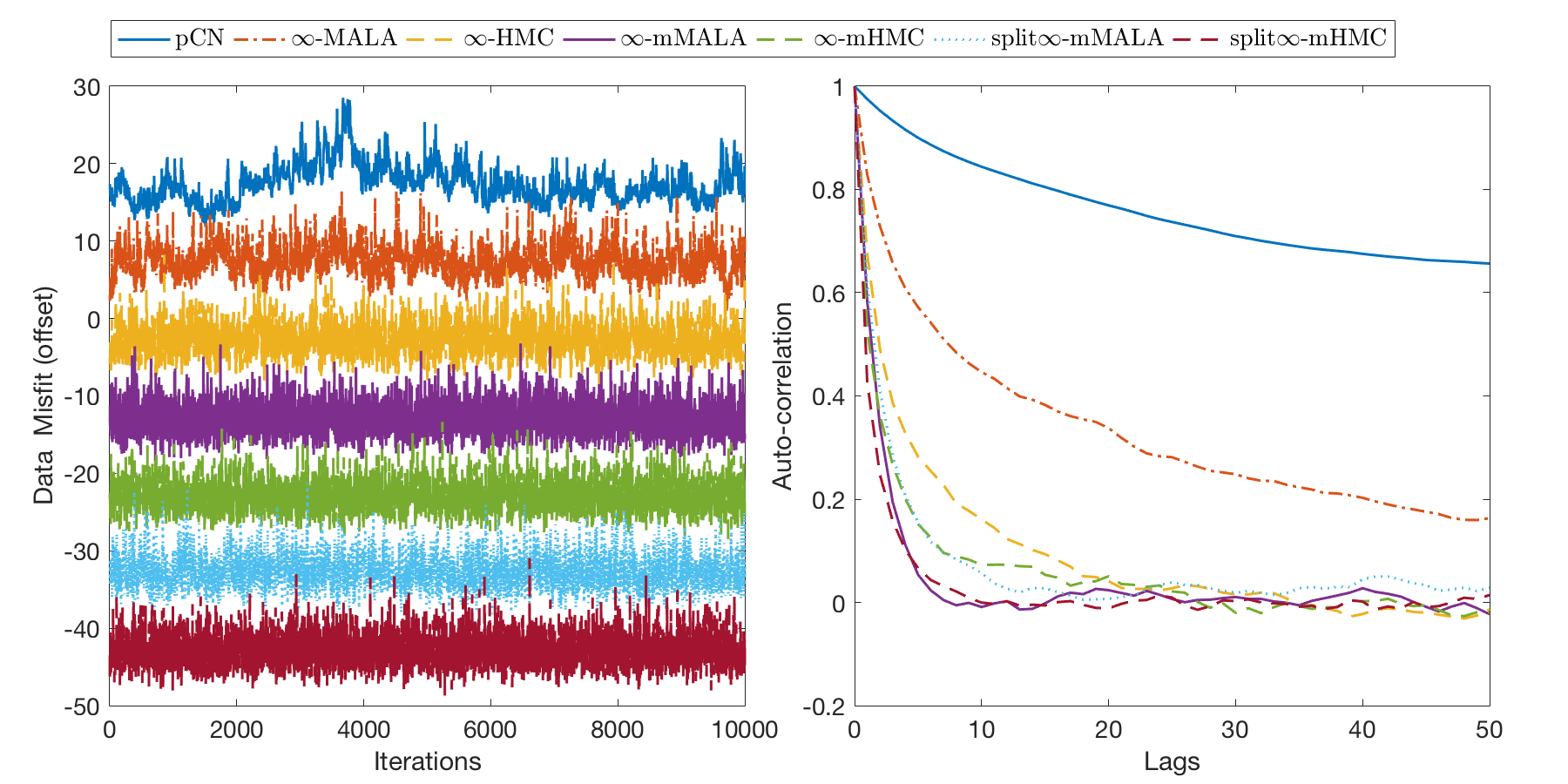}
  \vspace{-0.8cm}
  \caption{Groundwater flow problem: the trace plots of data-misfit function  (left panel, values have been offset for better comparison) and the 
 corresponding acf plots (right panel).\vspace{0.2cm}}
  \label{fig:misfit_acf_groundwater}
\end{figure}
%
\begin{table}[!h]
\tiny
\centering
\begin{tabular}{l|cclccc}
  \hline
Method & AP & s/iter & ESS(min,med,max) & minESS/s & spdup & PDEsolns \\ 
  \hline
pCN & 0.69 & 4.86E-03 & (5.72,17.23,52.6) & 0.118 & 1.00 & 11001 \\ 
  $\infty$-MALA & 0.71 & 2.23E-01 & (27.15,58.44,138.93) & 0.012 & 0.10 & 22002 \\ 
  $\infty$-HMC & 0.77 & 5.62E-01 & (302.37,461.03,590.36) & 0.054 & 0.46 & 54822 \\ 
  $\infty$-mMALA & 0.75 & 8.09E-01 & (1422.11,1747.68,2051.5) & 0.176 & 1.49 & 2222202 \\ 
  $\infty$-mHMC & 0.62 & 1.99E+00 & (2514.45,3667.88,4438.35) & 0.126 & 1.07 & 5562070 \\ 
  Split $\infty$-mMALA & 0.67 & 3.20E-01 & (654.22,1078.15,1283.37) & 0.205 & 1.74 & 572052 \\ 
  Split $\infty$-mHMC & 0.67 & 8.02E-01 & (3641.2,5230.48,5746.96) & 0.454 & 3.85 & 1434940 \\ 
   \hline
\end{tabular}
\caption{Sampling efficiency in the groundwater flow problem. Column labels are as follows.
AP: average acceptance probability; s/iter: average seconds per iteration; ESS(min,med,max): minimum, median, maximum of Effective Sample Size across all posterior coordinates; min(ESS)/s: minimum ESS per second;
spdup: speed-up relative to base pCN algorithm;
PDEsolns: number of PDE solutions during execution.
\vspace{0.2cm}}
\label{tab:groundwater}
\end{table}
Table \ref{tab:groundwater} compares the sampling efficiency of the various algorithms.
Once more information is introduced (gradient or/and  location-specific scales in the geometric methods) the mixing of the algorithms improves drastically. Even when the increased computational cost is taken under consideration, the overall effectiveness of Split $\infty$-mHMC, as measured by the minimal effective sample size (ESS) per CPU time (in secs), points to close to 4-fold improvement compared to pCN. In this example, the non-geometric methods $\infty$-MALA, $\infty$-HMC perform worse than pCN due to insufficient mixing improvement when weighted against the extra computations. 
The same holds for $\infty$-mHMC, clearly motivating in this case the significance of  the truncation technique for reducing computational costs within Split $\infty$-mHMC.

\begin{figure}[!h]
  \centering
  \includegraphics[width=1\textwidth]{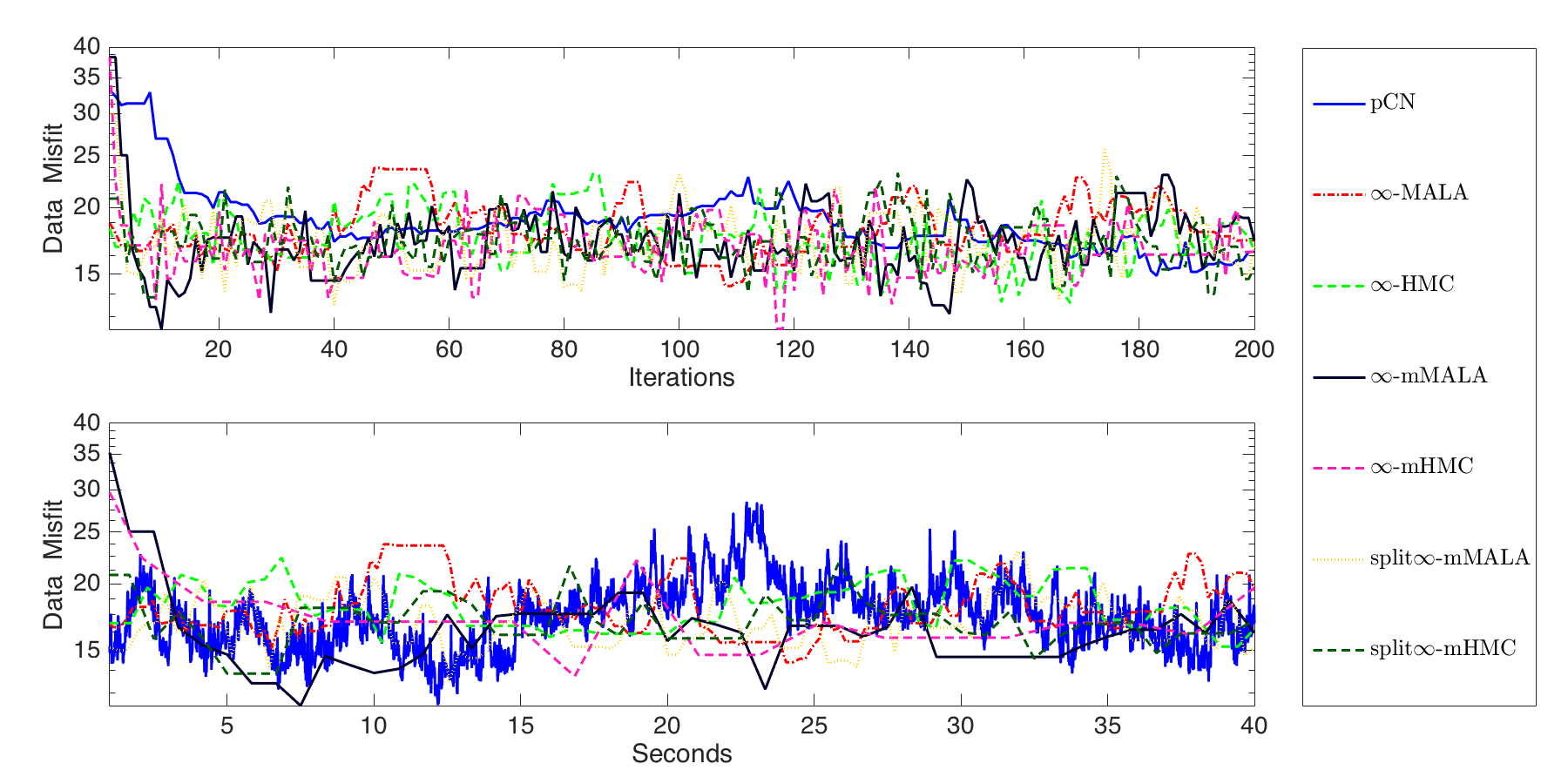}
  \vspace{-0.9cm}
  \caption{Groundwater flow problem: trace plots of data-misfits before burn-in for the first 200 iterations (upper panel) and first 40 seconds (lower panel) respectively.}
  \label{fig:misfit_groundwater}
\end{figure}

Figure \ref{fig:misfit_groundwater} shows the first few data-misfit evaluations
at the beginning of the algorithms. PCN  exhibits strong diffusive behavior. 
The lower panel, where the horizontal axis corresponds to execution time, 
seems to indicate that maybe the various methods are not dramatically better than 
pCN in this case.
Still, as mentioned above, the optimal speed-up against pCN  is by a factor of $4$.
In the two subsequent, more complex, examples the speed-up 
factor will be much larger.
Splitting methods with truncation number different from $D_0=25$ are also implemented: smaller $D_0$ causes the truncated Fisher operator to lose useful information while larger $D_0$ negatively impacts the computational advantage. One can refer to Figure \ref{fig:est4indep_groundwater} for illustration. Other results are omitted for brevity of exposition.

\begin{figure}[!h]
  \centering
  \includegraphics[width=1\textwidth]{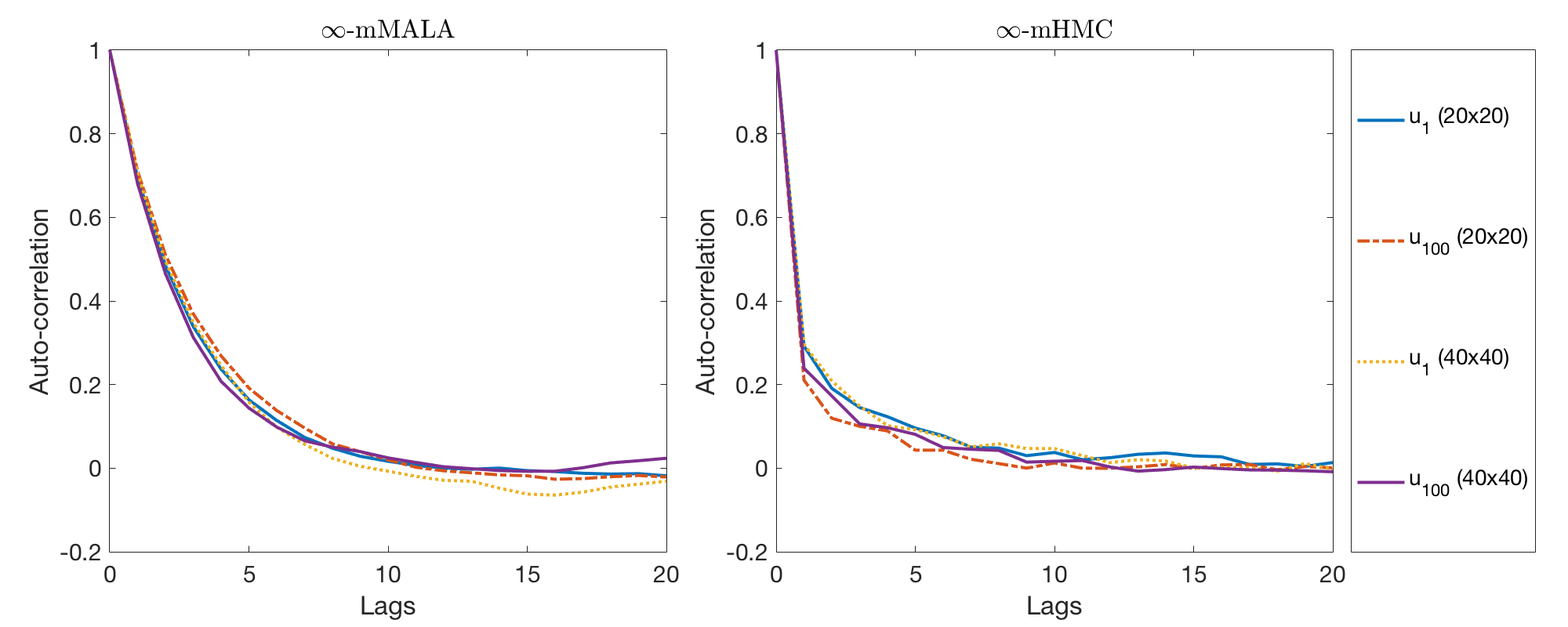}
\vspace{-1cm} 
  \caption{Groundwater flow problem: auto-correlation function of selected samples ($u_1, u_{25}, u_{100}$) generated by $\infty$-mMALA (left) and $\infty$-mHMC (right) with forward PDE solving carried on $20\times20$ mesh and $40\times40$ mesh.}
  \label{fig:acf4indep_groundwater}
\end{figure}

To verify mesh-independence of the proposed methods, we re-do the above inference with forward PDE solved on a refined, $40\times40$ mesh. Since the mesh-independence of non-geometric methods has been established in the literature \citep{beskos08,cotter13,beskos11}, and split algorithms are special cases of their full versions, we only compare the performance of $\infty$-mMALA (and $\infty$-mHMC) with PDE solved on $20\times20$ mesh and $40\times40$ mesh.
For $\infty$-mMALA, the two implementations share the same acceptance rate $75\%$ and their effective sample sizes (minimum, median, maximum) are $(1422.11,1747.68,2051.5)$ and $(1263.78,1757.22,2056.68)$ respectively. For $\infty$-mHMC, the two implementations have similar acceptance rates ($62\%$ and $61\%$ repectively), and effective sample sizes $(2514.45,3667.88,4438.35)$ and $(2311.26,3469.34,4469.44)$ respectively.
Figure \ref{fig:acf4indep_groundwater} illustrates that for both $\infty$-mMALA and $\infty$-mHMC, the auto-correlation functions of selected samples decay with lag but do not deteriorate under mesh refinement. This fact means that the number of MCMC steps to reach equilibrium is independent of the mesh \citep{hairer14}.
Figure \ref{fig:est4indep_groundwater} shows the close posterior mean estimates of the permeability field by $\infty$-mHMC with PDE solved on those two meshes (Similar result exists for $\infty$-mMALA but is omitted), which also
qualitatively confirms the mesh-independence of $\infty$-mMALA and $\infty$-mHMC. The column wise comparison of estimates using different number of modes indicates that most posterior information is concentrated in the subspace formed by the first 25 eigen-directions.

\begin{figure}[t]
  \centering
  \includegraphics[width=1\textwidth]{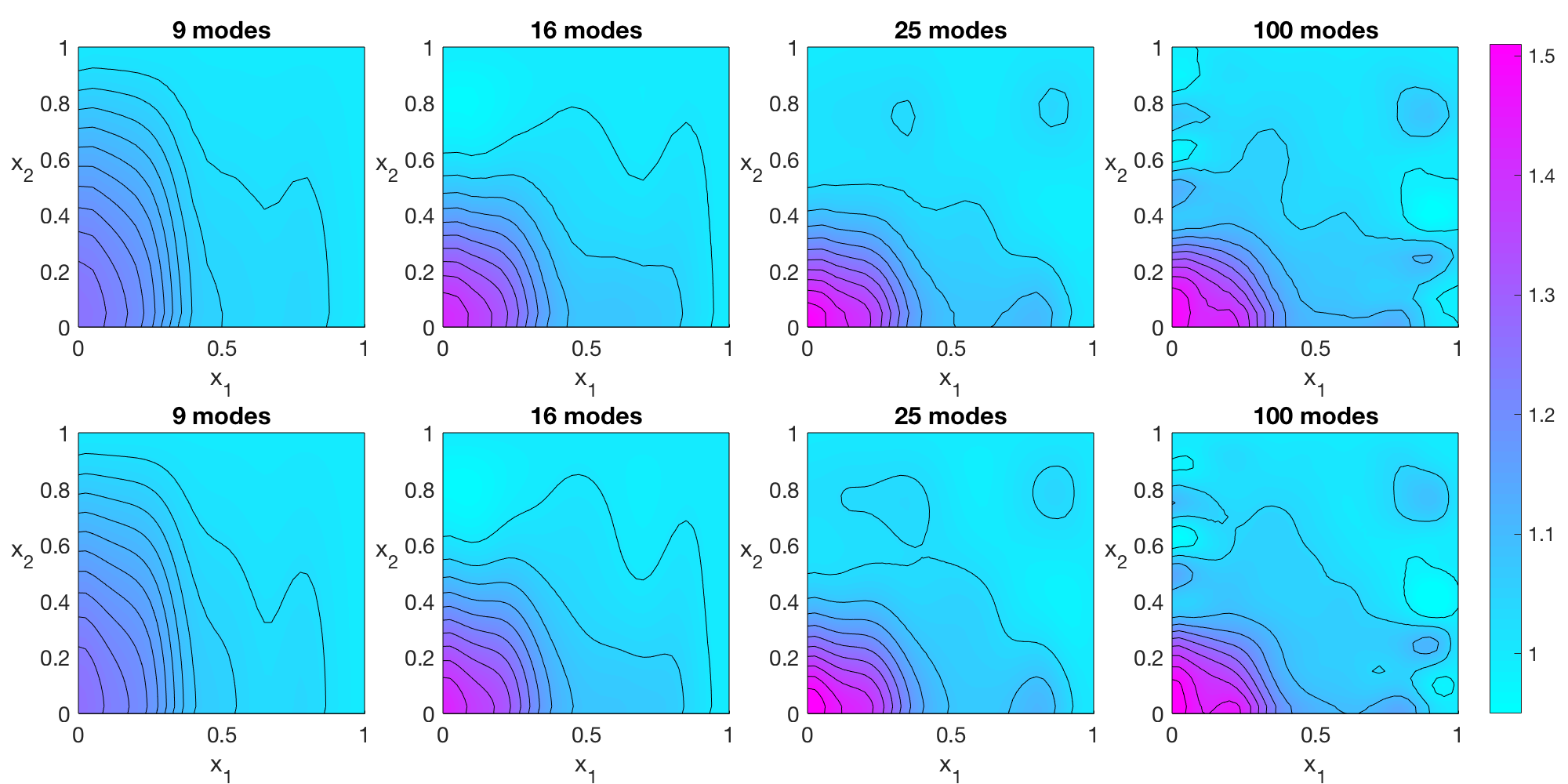}
\vspace{-1cm} 
  \caption{Groundwater flow problem: estimated permeability field $e^u$ using samples by $\infty$-mHMC with forward PDE solving carried on $20\times20$ mesh (upper row) and on $40\times40$ mesh (lower row). Each column corresponds to estimates with different number of modes (components of $\{u_i\}$).}
  \label{fig:est4indep_groundwater}
\end{figure}

\begin{figure}[t]
  \begin{center}
    \includegraphics[width=1\textwidth]{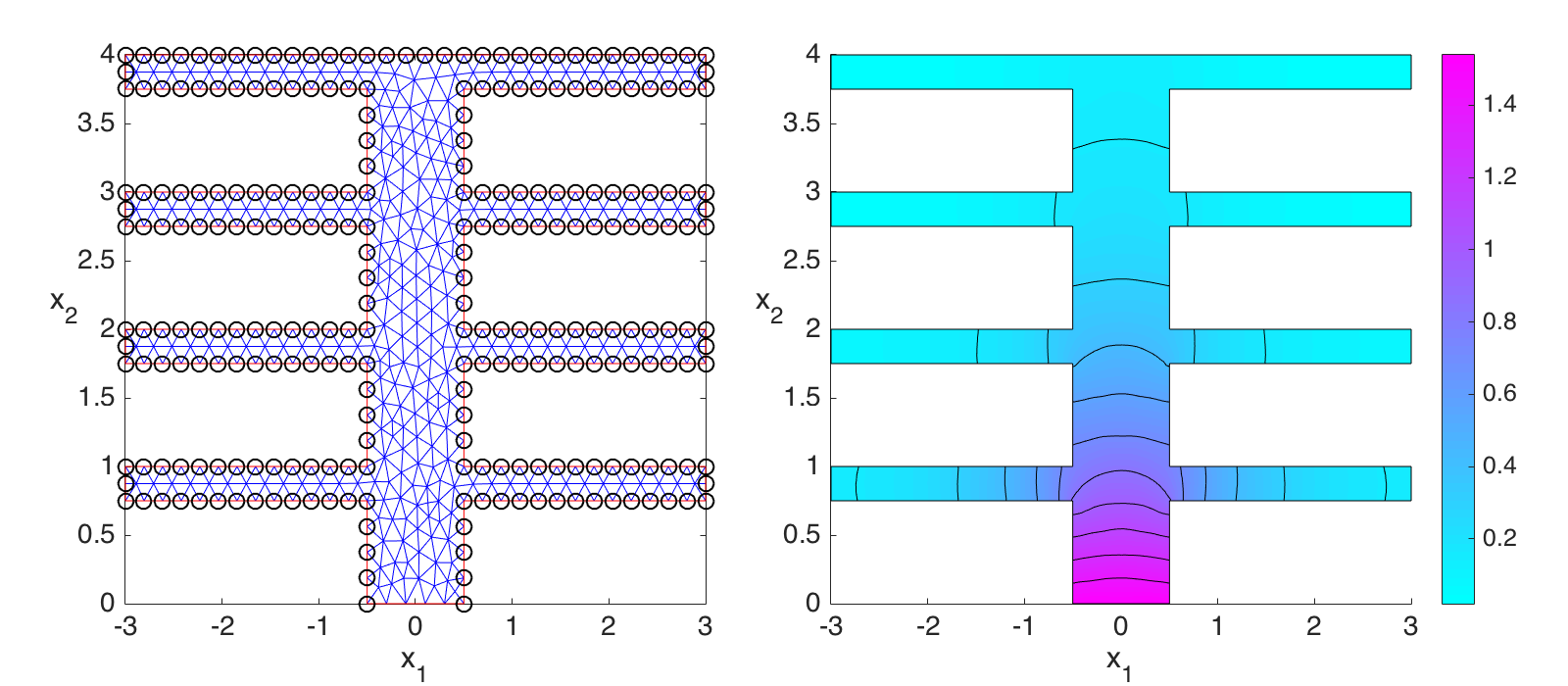}
  \end{center}
  \vspace{-0.8cm}
  \caption{Thermal fin problem: the location of observations (circles) (left panel) and the forward PDE solution $p^{\dagger}$ under the true 
  parameter $u^{\dagger}$ (right panel).}
  \label{fig:obs_soln_thermalfin}
\end{figure}

\subsection{Thermal Fin}
We now consider the following thermal fin model:
\begin{equation}\label{ThermalFin}
\begin{aligned}
-\nabla\cdot (e^{u({\bf x})}\nabla p({\bf x})) = 0\ ,\qquad\qquad \,\,\quad &\mathbf{x}\in\mathcal E^{0}=\mathrm{Interior}(\mathcal{E})\, , \\
(e^{u({\bf x})}\nabla p({\bf x})) \cdot {\bf n}  = -Bi\,\cdot p({\bf x})\ ,\quad &\mathbf{x}\in\pa\mathcal E\backslash\Gamma\ , \\
(e^{u({\bf x})}\nabla p({\bf x})) \cdot {\bf n} = 1\ ,  \quad\qquad\qquad\, &\mathbf{x}\in\Gamma=[-0.5,0.5]\times\{0\}
\ . 
\end{aligned}
\end{equation}
These equations model the heat conduction over 
the non-convex domain $\mathcal E$ depicted in Figure \ref{fig:obs_soln_thermalfin},
where $\Gamma=[-0.5,0.5]\times\{0\}$ is a part of the boundary $\pa\mathcal E$ on which the inflow heat flux is 1.
For the rest of the boundary we assume Robin boundary conditions.
Following \cite{bui15}, we set the Biot number to $Bi=0.1$.
The forward problem (\ref{ThermalFin}) provides the temperature $p$ given the heat conductivity function $e^u$ and the inverse problem involves reconstructing $u$ from noisy observations of $p$.
The complexity of the model domain
makes this inverse problem more challenging than the previous groundwater flow problem.

The prior for $u$ is obtained as explained at Subsection 
\ref{sec:gauss}, for domain $\mathcal{D}=[-3,3]\times [0,4]$.
We have chosen a rectangular domain $\mathcal{D}$ for $u$ which
contains the domain $\mathcal{E}$ of the PDE 
as a convenient way to construct the prior. 
However it should be mentioned that such a construction may introduce non-physical correlations between the fins;
priors which are geometry-adapted could be used but would be more
complicated to implement and maybe go beyond the scope of this paper.
In this example, we set $\alpha=0$, $s=1.2$, 
$\sigma^2=1$ in the specification of $\Co$. 
The true log-conductivity field 
$u^{\dagger}$ has coordinates
$u_i^{\dagger}=\lambda_i^{1/2}\,\sin((i_1-\tfrac{1}{2})^2+(i_2-\tfrac{1}{2})^2)\cdot 
\delta\,[\,i_1\le 10,\,i_2\le 10\,]$
and the simulated data are obtained by 
solving \eqref{ThermalFin} on a triangular mesh (left panel of Figure \ref{fig:obs_soln_thermalfin}) with discretization step-size $h_{max}=0.1$.
Then, $N=262$ observations are taken along the Robin boundary $\pa\mathcal E\backslash\Gamma$ (we denote the positions of the observations 
$\{\mathbf{x}_n\}$, $1\le n \le N$)
 and contaminated with Gaussian noise with mean zero and standard deviation
$\sigma_y=0.01\cdot \max_{1\leq n\leq N}\{p^{\dagger}({\bf x}_n)\}$, as in \citep{bui15}.
When running the MCMC algorithms, 
we project on the coordinates on 
$\{i_1,i_2\le 10\}$, and use the same finite
element construction as above.
HMC algorithms use a number of leapfrog steps randomly chosen between 1 and 4. 
The split methods apply the geometric principle 
on $\{i_1,i_2\le 5\}$.
Thus similarly as the previous example, $|I_0|=100$ and $D_0=25$.
\begin{figure}[t]
  \centering
  \includegraphics[width=1\textwidth]{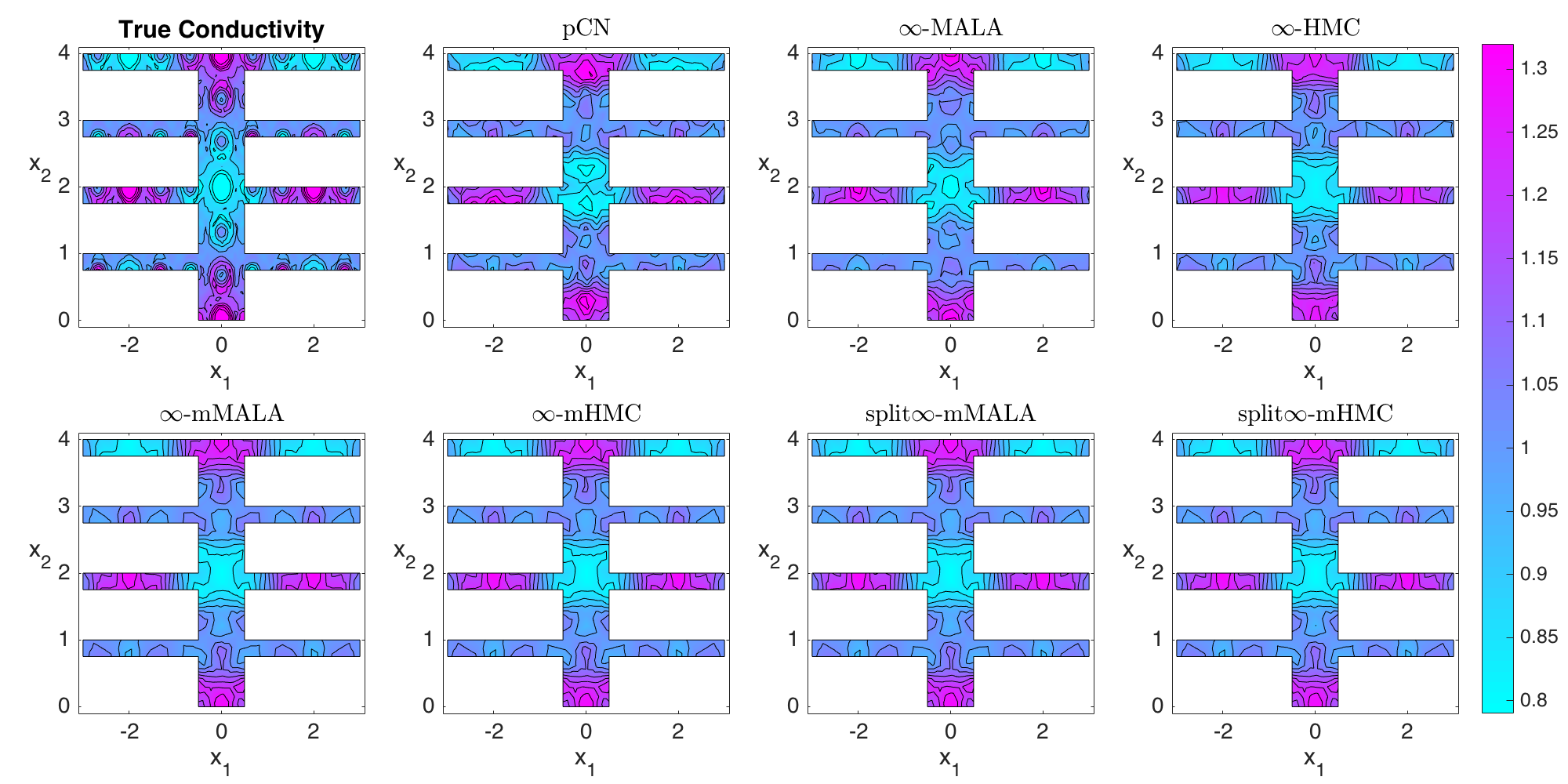}
  \vspace{-0.9cm}
  \caption{Thermal fin problem: the true
  heat conductivity field $e^{u^{\dagger}}$ (upper left-most) and the posterior mean estimates obtained by the various MCMC algorithms.
  \vspace{0.2cm}}
  \label{fig:truth_est_thermalfin}
\end{figure}

In this example, there are ample data points (262) to provide enough information
in inferring (100) unknown parameters, which is different from the previous example
as an underdetermined elliptic inverse problem (inferring 100 unknown parameters from 33 data points) \citep{dashti15}.
As shown in Figure \ref{fig:truth_est_thermalfin},
the posterior mean estimates of heat conductivity are consistent across different algorithms and 
close to the truth. Due to having more informative data in this 
example, the posterior mean is closer to the truth 
than in the previous example (see Figure \ref{fig:truth_est_groundwater}).
%
\begin{table}[!ht]\tiny
\centering
\begin{tabular}{l|cclccc}
  \hline
Method & AP & s/iter & ESS(min,med,max) & minESS/s & spdup & PDEsolns \\ 
  \hline
pCN & 0.67 & 6.97E-03 & (3.61,8.67,29.93) & 0.052 & 1.00 & 11001 \\ 
  $\infty$-MALA & 0.70 & 9.60E-02 & (5.52,15.07,33.91) & 0.006 & 0.11 & 22002 \\ 
  $\infty$-HMC & 0.75 & 2.34E-01 & (24.78,81.13,156.41) & 0.011 & 0.20 & 55264 \\ 
  $\infty$-mMALA & 0.79 & 5.12E-01 & (1729.28,2224.8,2474.28) & 0.338 & 6.51 & 2222202 \\ 
  $\infty$-mHMC & 0.69 & 1.31E+00 & (4018.07,5679.26,6956.14) & 0.306 & 5.90 & 5582270 \\ 
  Split $\infty$-mMALA & 0.77 & 1.53E-01 & (1180.78,1792.34,2026.81) & 0.770 & 14.87 & 572052 \\ 
  Split $\infty$-mHMC & 0.72 & 3.85E-01 & (5327.64,7107.08,8335.14) & 1.384 & 26.70 & 1432704 \\ 
   \hline
\end{tabular}
\caption{Sampling efficiency in the thermal fin problem. Column labels are as in Table \ref{tab:groundwater}.
\vspace{0.2cm}}
\label{tab:thermalfin}
\end{table}
\begin{figure}[!h]
  \centering
  \includegraphics[width=1\textwidth]{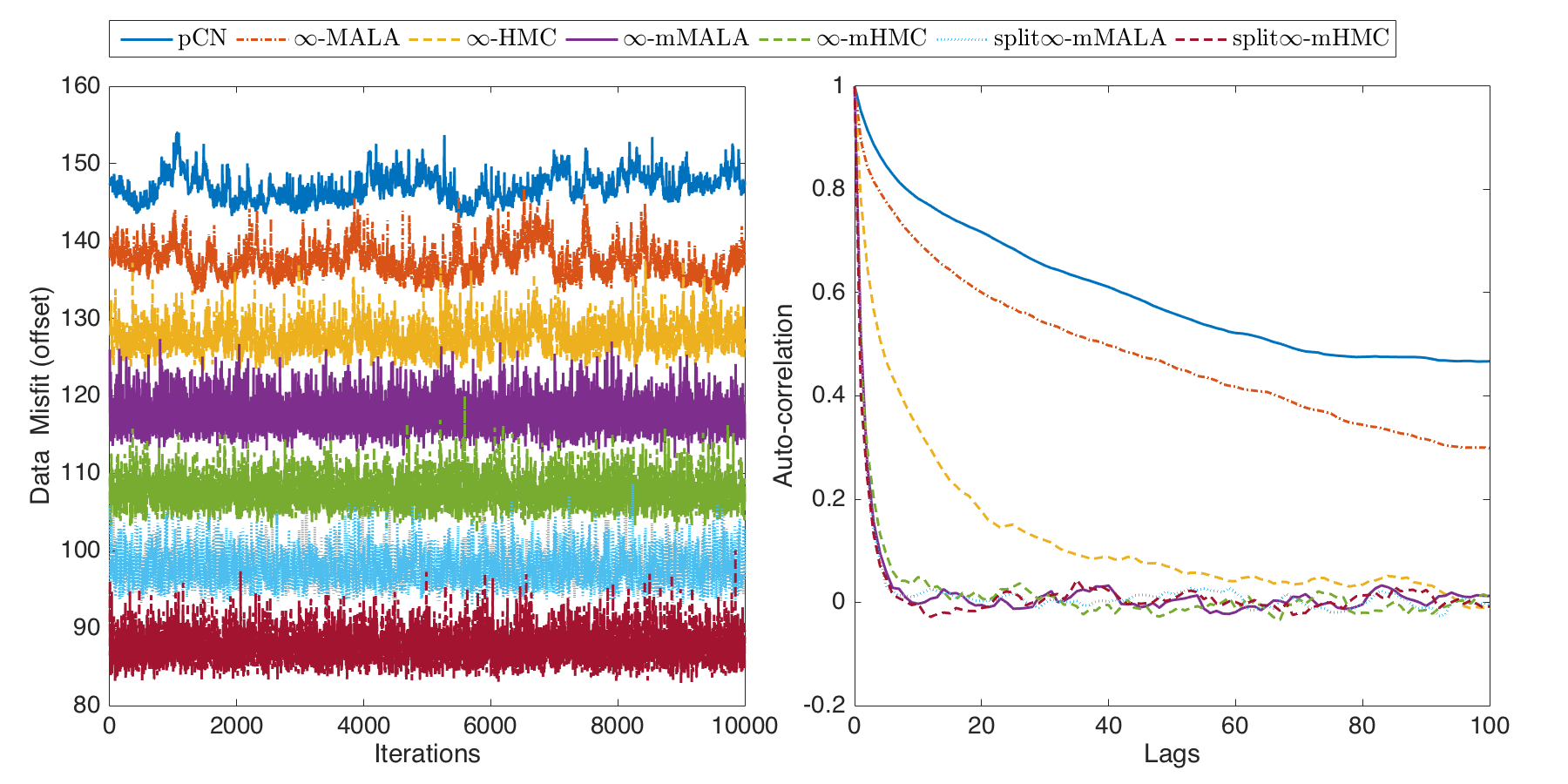}
  \vspace{-0.9cm}
  \caption{Thermal fin problem: trace plots of data-misfit function (left panel, values have been offset for better comparison) and
  the corresponding acf functions (right panel).}
  \label{fig:misfit_acf_thermalfin}
\end{figure}
Table \ref{tab:thermalfin} and Figure \ref{fig:misfit_acf_thermalfin} compare the sampling efficiency of different algorithms. 
Notice that more than an order of magnitude of improvement is observed for Split $\infty$-mMALA and Split $\infty$-mHMC compared to pCN. 
%
%
\begin{figure}[!h]
  \centering
  \includegraphics[width=1\textwidth]{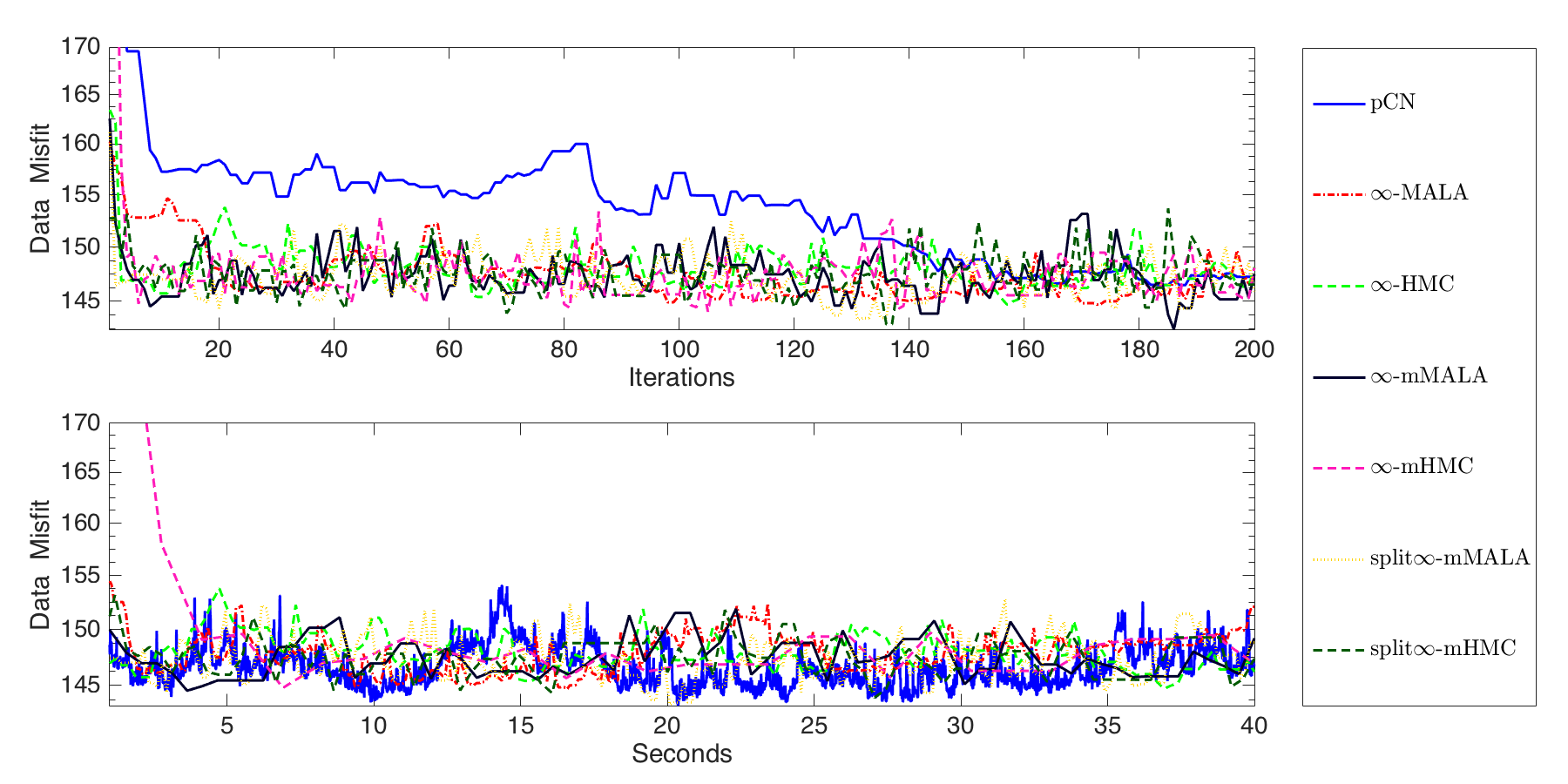}
  \vspace{-0.9cm}
  \caption{Thermal fin problem: trace plots of data-misfits before burn-in for the first 200 iterations (upper) and the first 40 seconds (lower) respectively.}
  \label{fig:misfit_thermalfin}
\end{figure}
In Figure \ref{fig:misfit_thermalfin}, pCN needs several iterations to reach the stationary stage. 
Notice that in this case also $\infty$-mHMC requires some time 
before reaching the stationary regime.

\subsection{Laminar Jet}
We consider the 2D incompressible Navier-Stokes equation:
\begin{equation}\label{laminarjet}
\begin{aligned}
\textrm{Momentum}:\quad &
- \div \left( \nu\,\big(\nabla \bu + \nabla \bu^\top
\big) \right) + \bu\cdot \nabla \bu + \nabla p = 0\ , \\
\textrm{Continuity}:\quad &
\div\, \bu = 0\ ,\\
& \bu \cdot \bn = -\theta(y)\ , \quad \bsigma_n \times \bn = 0\ ,  \quad {\rm on}\,\,\; \mathcal I\ ,  \\
& \bsigma_n + \beta\,(\bu\cdot \bn)_{-}\, \bu = \bZero\ ,  \quad {\rm on}\,\,\; \mathcal O\ ,  \\
& \bu \cdot \bn = 0\ , \; \bsigma_n \times \bn = 0\ ,   \quad {\rm on}\;\,\, \mathcal B\ ,  
\end{aligned}
\end{equation}
where $\bu=(u_1,u_2)$ is the velocity, $p$ is the pressure and
 $\nu>0$ is the viscosity.
Vector $\bn$ denotes the unit normal to the mesh boundary and $$\bsigma_n = - p\,\bn + \nu\, (\nabla \bu + \nabla \bu^\top) \cdot \bn$$ represents the boundary traction. Also, 
$$(\bu\cdot \bn)_{-}\, = (\bu \cdot \bn - | \bu \cdot \bn |)/2$$ and $\beta \in (0, 1]$ is the backflow stabilization parameter
in \cite{moghadam11}.
This PDE models non-reacting turbulent jet dynamics.
$\mathcal{I}, \mathcal{O}, \mathcal{B}$ denote the inlet, outlet and
bounding sides respectively, to be described below.

\begin{figure}
  \centering
  \includegraphics[width=1\textwidth]{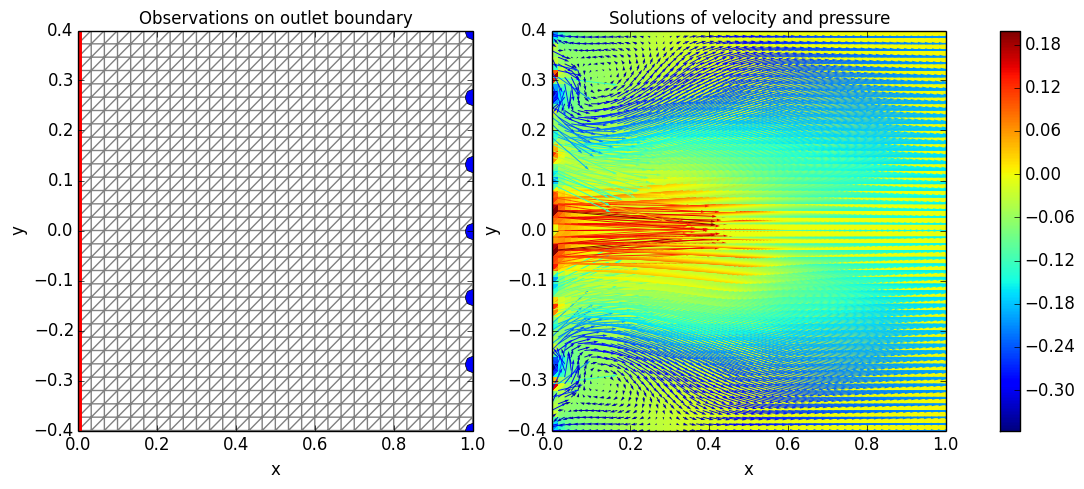}
  \vspace{-0.9cm}
  \caption{Laminar jet problem: (left panel) the location of inlet velocity to be inferred (red line) and the measurement
locations (blue dots); (right panel) the forward PDE solution with true unknown $\theta^{\dagger}$, with the heat map showing the pressure $p$ and the arrows representing the velocity field $\bu$.}
  \label{fig:obs_solns_laminar}
\end{figure}

We will describe a concrete simplified problem setting following \cite{klein03}.
The relevant domain $\mathcal{E}$ for the PDE is a rectangle with length $L_x=10L$ and width $L_y=8L$, 
with parameter $L$ being a typical lengthscale of the
(unknown) inlet velocity field; it is set to $L=0.1$ in this experiment. 
The induced domain $\mathcal{E}=[0,1]\times [-0.4,0.4]$ is shown 
on the left panel of Figure  \ref{fig:obs_solns_laminar}. 
We consider the following boundary conditions. At the inlet boundary $\mathcal I=\{\,x=0,\,y \in (- L_y/2, L_y/2)\,\} $
we prescribe a normal velocity profile $\theta(y)$ and vanishing tangential stress. 
At the outflow boundary  $\mathcal O=\{x=L_x\ , y \in (-L_y/2\,, L_y/2)\}$ we prescribe a traction-free condition plus an additional convective traction term to stabilize regions of possible backflow \citep{moghadam11}. 
Finally, on the bounding sides 
$\mathcal B=\{x\in(0, L_x),\,y=\pm L_y/2\}$ we prescribe free-slip conditions.
A typical solution is shown in the right panel of Figure \ref{fig:obs_solns_laminar}, where the heat map shows the pressure $p$ and the arrows 
represent the velocity field $\bu$.
Note that the color change along the inlet boundary reflects the persistence of high frequencies in the true inflow velocity profile (see also the left panel of Figure \ref{fig:truth_est_laminar}). 

\begin{figure}
  \centering
  \includegraphics[width=1\textwidth]{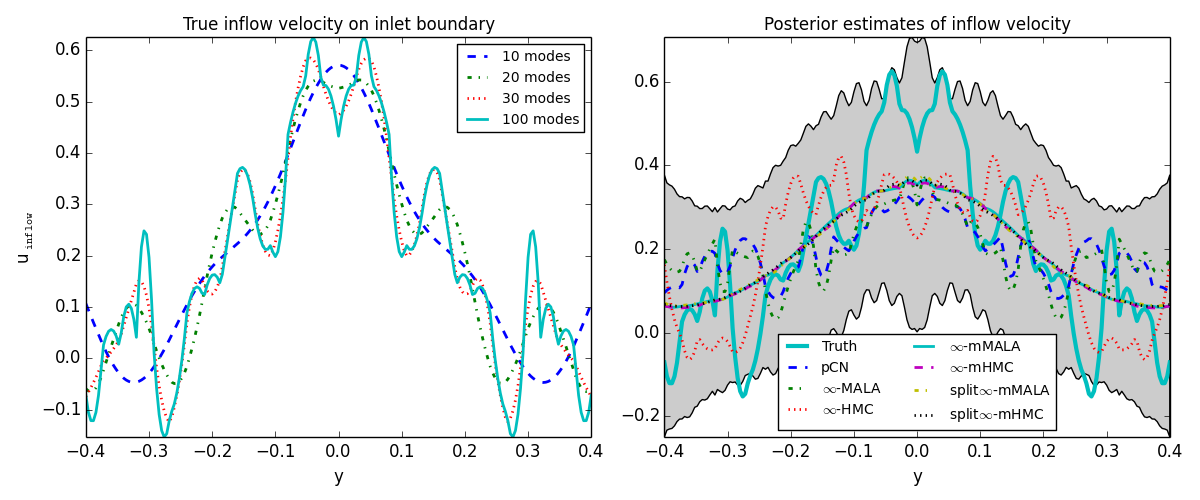}
  \vspace{-1cm}
\caption{Laminar jet problem: true inflow velocity  for increasing number of frequencies (left panel); the true $\theta^{\dagger}$ used 
corresponded to the highest number shown, 100. 
Also, posterior mean estimates provided  by
the MCMC algorithms (right panel). 
Results of all geometric algorithms (with small `m') agree with each other and the others (non-geometric methods) do not because they have not converged.
The shaded region shows the $95\%$ credible band constructed with samples from $\infty$-mHMC.}
  \label{fig:truth_est_laminar}
\end{figure}

Given an inflow velocity profile $\theta=\theta(y)$ on $\mathcal I$, the forward problem 
computes $\bu(x,y)$, and $\varphi(y) = \bu(L_x,y)$.
The inverse problem aims to infer $\theta=\theta(y)$ given noisy observations of $\varphi(y)$ on the right boundary
$\mathcal O$.
We assume an 1D Gaussian prior on the super-domain $[-1,1]$ 
as explained in Subsection \ref{sec:gauss}.
%
We choose hyper-parameters $\sigma=0.5$, $\alpha=1$ and $s=0.8$. 
We obtain the true path $\theta^{\dagger}$
by sampling the coefficients $\theta_i^{\dagger}$,
$1\le i \le 100$, from the prior with $\theta_i^{\dagger}=0$, $i>100$.
The true inflow velocity $\theta^{\dagger}$ on $[-0.4,0.4]$ is shown at the left panel of Figure \ref{fig:truth_est_laminar}.
Note here negative values of $\theta(y)$ (around $y=\pm0.4, \pm0.3$) indicate backward flow, which also can be seen in the right panel of Figure \ref{fig:obs_solns_laminar}.
We solve the Laminar equation for $\nu=3\times10^{-2}$, $\beta=0.3$ on a $60\times60$ mesh 
and obtain 7 observations from the velocity field at the locations indicated by blue dots on the left panel of Figure \ref{fig:obs_solns_laminar}, contaminated with Gaussian noise of variance $\sigma^2_{obs}=10^{-4}$.
We stress here that this is a complex inverse problem due to the non-linearity of the forward PDE and the sparsity of observations. Each forward solution relies on an expensive Newton iteration 
with no clear theory about convergence of solutions 
when using different initializations.
In this experiment, 
we choose the viscosity $\nu=3\times10^{-2}$ as a compromise between reasonable convergence
rate in the nonlinear solver, which favors larger $\nu$, and obtaining
interesting flow structure, which favors smaller $\nu$.
We also adopt the perspective of using a fixed 
initial position ($\theta_i=0$ for all $i$ here) for the Newton iteration 
every time the PDE  dynamics are invoked, 
so that there is a well-defined map (on a given grid) from $\theta(y)$
to the likelihood of the observations. 
The required adjoints for gradient and metric-action (metric-vector product) are linear, and hence not too expensive to compute. 
The backflow stabilization term (in the 4th equation of \eqref{laminarjet}) involves taking
the minimum of $\bu\,\cdot\, \bn$ with $0$. This term is non-differentiable wherever $\bu\cdot \bn = 0$, and
thus the unknown-to-likelihood map is formally non-differentiable 
on the set $\{{\bf x}\in \mathcal O\,;\, \bu({\bf x})\cdot \bn = u_1({\bf x}) = 0\}$.
In future work we hope to extend geometric methods to such semi-smooth maps.
However, we believe that this non-smoothness occurs on sets of
measure zero in parameter space for the chosen PDE configuration, and hence poses
no difficulties in practice when computing derivatives in  geometric MCMC.

We run the various MCMC algorithms 
(all initialized at zero) for $1.1\times10^4$ iterations, treating the first $10^3$ samples as burn-in. The posterior is obtained 
by stopping the K-L expansion for the prior at $|I_0|=100$ and solving the PDE on a $30\times30$ mesh.
The split-methods used location-specific scales 
up to $D_0=30$.
HMC algorithms use a number of leapfrog steps randomly chosen between 1 and 4.
%
We mention here an important practical consideration that arises
when solving this problem. For almost all proposed states within MCMC
the Newton solver converged. However with very low probability,
and in almost all the experiments we ran, situations arise in which
the proposed MCMC states led to divergence of the Newton solver. Whilst
this might be ameliorated to some extent by different initializations
of the Newton method, for reasons described above we have fixed the
initialization. We deal with the divergence of Newton method in
these situations by rejecting such proposals with probability $1$, 
i.e. we remove these low probability states from the domain of the posterior.
%
%
%

Unlike the previous two PDE examples,
none of the non-geometric methods 
converged to equilibrium  
due to requiring very small step-sizes 
($\mathcal{O}(10^{-4})$) to provide 
non-negligible acceptance rates. 
The right panel of Figure \ref{fig:truth_est_laminar},
shows the posterior means as estimated 
by the various MCMC algorithms. 
As expected, the estimate does not match 
the true inflow velocity $\theta^{\dagger}$ 
in the high frequencies due to limited amount of data. 
%
Note that the $95\%$ credible band calculated with samples from $\infty$-mHMC is wide and covers most of the true inflow velocity (solid cyan line).
\begin{figure}[!h]
  \begin{center}
     \includegraphics[width=1\textwidth]{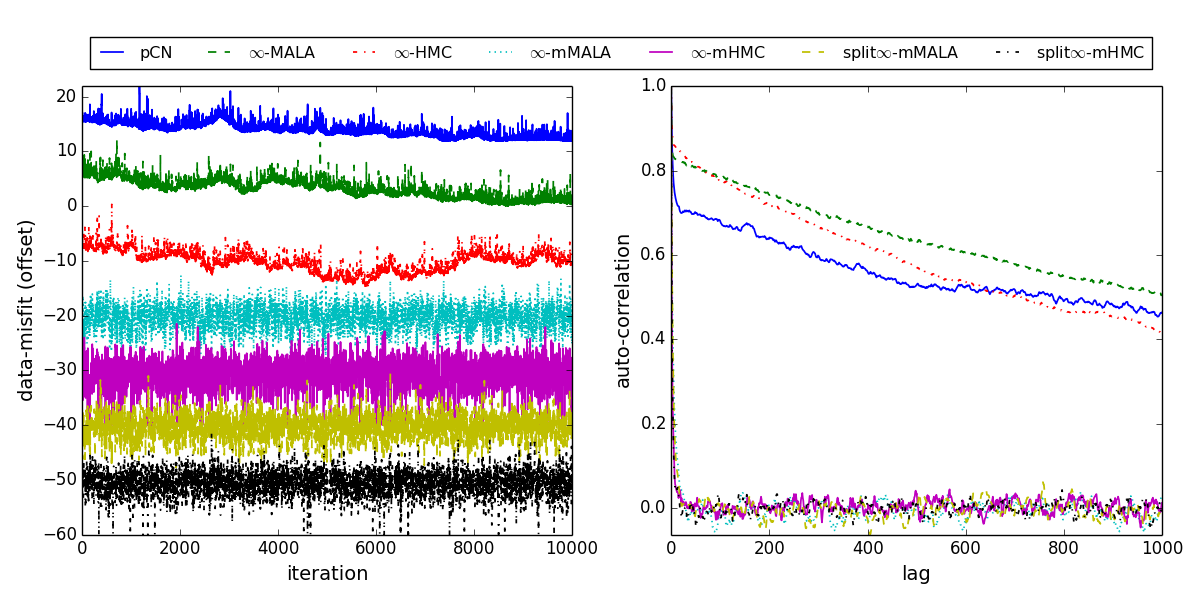}
  \end{center}
  \vspace{-0.8cm}
  \caption{Laminar jet problem:  trace plots of the data-misfit function  (left panel, values have been offset for better comparison) and the corresponding acf functions (right panel).}
  \label{fig:misfit_acf_laminar}
  \vspace{0.2cm}
\end{figure}
Figure \ref{fig:misfit_acf_laminar} illustrates the extremely high auto-correlation of samples in the case of the  non-geometric methods due to 
ineffective small step-sizes.
 The left panel indicates that 
 non-geometric methods have not converged and the right panel shows high auto-correlation even at a lag of 1000.
%
%
\begin{table}[!h]\tiny
\centering
\begin{tabular}{l|cclccc}
  \hline
Method & AP & s/iter & ESS(min,med,max) & minESS/s & spdup & PDEsolns \\ 
  \hline
pCN & 0.61 & 1.29 & (5.24,   6.66,  13.33) & 4.05E-04 & 1.00 & 22004 \\ 
  $\infty$-MALA & 0.66 & 1.68 & (5.38,  6.62, 19.53) & 3.21E-04 & 0.79 & 33005 \\ 
  $\infty$-HMC & 0.72 & 3.81 & (5.41,  7.43, 16.44) & 1.42E-04 & 0.35 & 82466 \\ 
  $\infty$-mMALA & 0.68 & 5.97 & (1075.24, 2851.22, 3867.08) & 1.80E-02 & 44.47 & 2233205 \\ 
  $\infty$-mHMC & 0.58 & 13.33 & (2058.42, 3394.17, 4560.03) & 1.54E-02 & 38.13 & 5575696 \\ 
  Split $\infty$-MMALA & 0.57 & 3.66 & (1079.55, 1805.89, 2395.13) & 2.95E-02 & 72.82 & 693065 \\ 
  Split $\infty$-mHMC & 0.60 & 6.88 & (2749.63, 3974.36, 5498.03) & 4.00E-02 & 98.67 & 1721694 \\ 
   \hline
\end{tabular}
\caption{Sampling Efficiency in the laminar jet problem. Column labels are as in Table \ref{tab:groundwater}.
}  
\label{tab:laminar}
\end{table}
\begin{figure}[t]
  \begin{center}
     \includegraphics[width=1\textwidth]{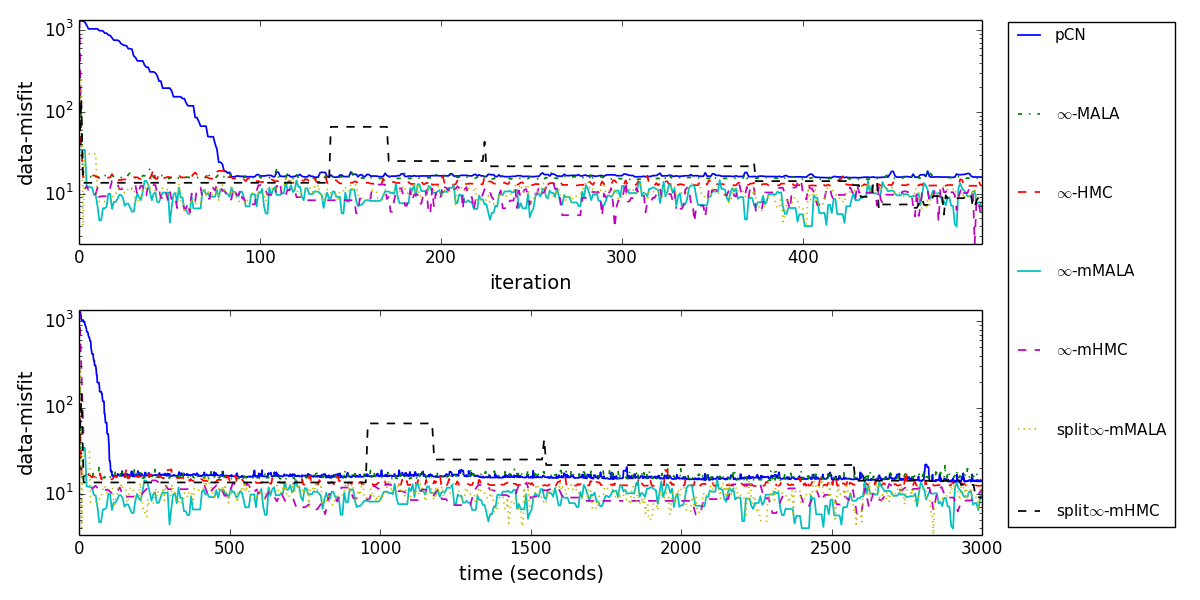}
  \end{center}
  \vspace{-0.7cm}
  \caption{Laminar jet problem: trace plots of data-misfits before burning-in for the first 500 iterations (upper panel) and the first 3000 seconds (lower panel) respectively.}
  \label{fig:misfit_laminar}
\end{figure}
Table \ref{tab:laminar} shows that the proposed geometric methods yield almost 2 orders of magnitude improvement in sampling efficiency compared with pCN.
Figure \ref{fig:misfit_laminar} illustrates the first few data-misfit values according to different sampling methods. 
The upper plot shows pCN and $\infty$-MALA 
have not reached the center of the posterior,
 while $\infty$-HMC starts to approach it after 400 iterations.
The lower plot verifies that this happens after 2500 seconds.
It is also interesting to note that unlike other geometric methods, split $\infty$-mHMC takes about 450 iterations and 3000 seconds to enter the convergent region.
All the above summaries confirm that geometric methods are advantageous in  sampling efficiency.

\begin{figure}[t]
  \begin{center}
     \includegraphics[width=.8\textwidth]{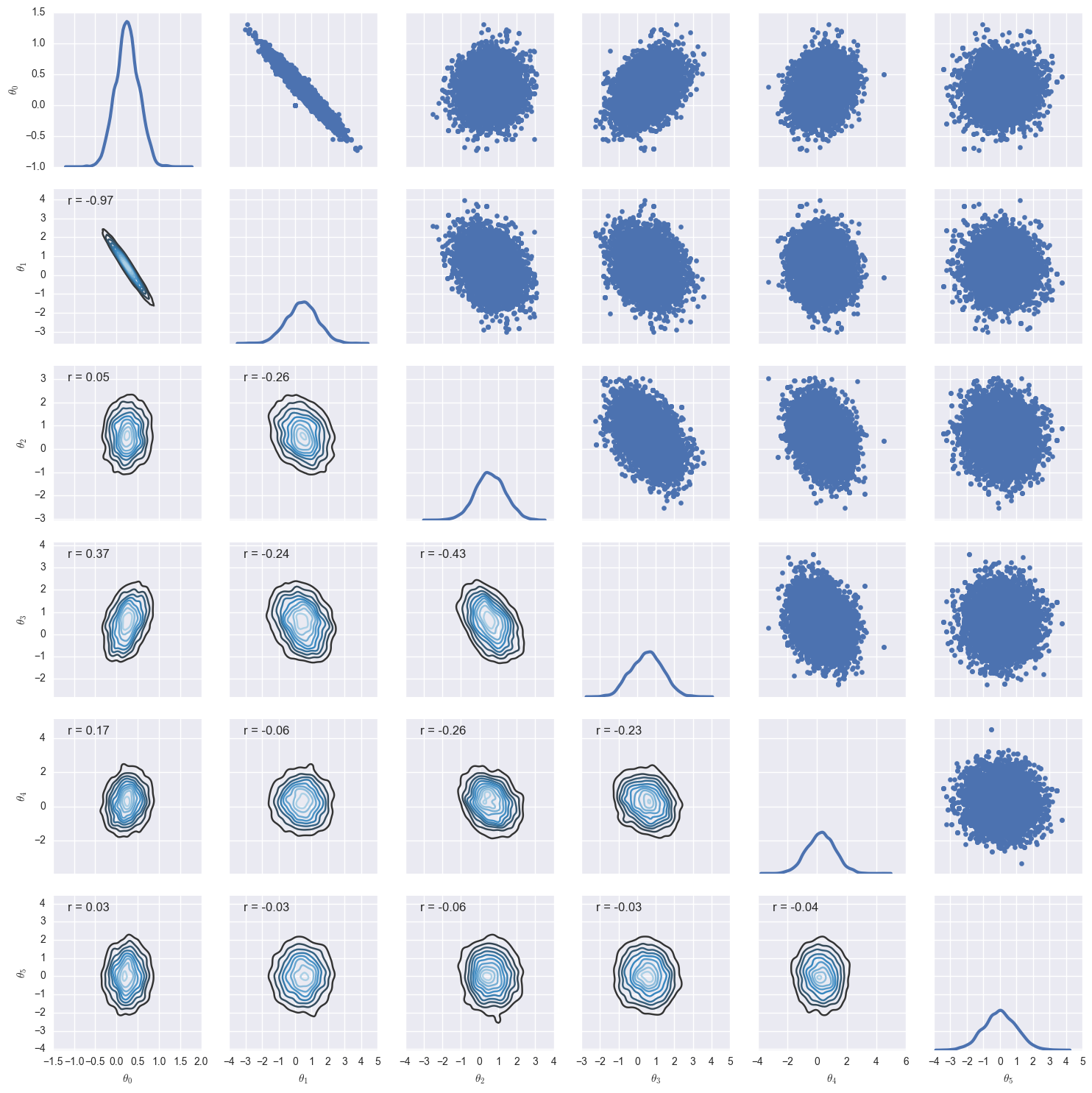}
  \end{center}
  \vspace{-0.7cm}
  \caption{Pair-wise marginal posterior distributions of the first 6 unknown frequencies of $\theta_1,\cdots,\theta_6$ in the laminar jet problem.}
  \label{fig:pairmarginal_laminar}
\end{figure}


\section{Conclusion and Discussion}\label{sec:conclusion}
This paper makes a number of contributions in the development
of MCMC methods appropriate for the solution of inverse problems
involving complex forward models with unknown parameters 
defined on infinite-dimensional Hilbert spaces. In particular: we generalize
the simplified Riemannian manifold MALA of \cite{girolami11} 
from finite to infinite dimensions, and develop an HMC-version 
of the new method; we establish a connection between these 
infinite-dimensional
geometric MALA and HMC algorithms; we develop a straightforward
dimension reduction methodology which renders the methods highly
effective in practice; we demonstrate the advantages of using HMC methods,
built around ballistic motion, i.e. move with large step-size, that suppresses random walk behavior.
All the algorithms are shown to be well-defined in the infinite dimensional setting,
and three numerical  studies demonstrate the effectiveness of the new
methodology.

Some recent works have investigated incorporating information about the posterior within MCMC algorithms of mesh-independent mixing times,
see e.g.\@
\cite{law2014} and the Dimension-Independent 
Likelihood-Informed MCMC in  \citep[DILI,][]{cui16}. However, 
these approaches aim to make use of the curvature of the posterior 
at a fixed position (typically, the MAP, i.e.\@ the maximiser of the posterior).
The geometric methods defined here can be more appropriate for 
distributions with more complex non-Gaussian structures.
In our laminar jet example for instance, Figure \ref{fig:pairmarginal_laminar}
illustrates the non-Gaussianity of the posterior, 
thus incorporation of information about the local geometry can be beneficial in this context.
Our methodology does not require pre-processing steps (e.g.\@ finding the MAP and the Hessian at the MAP).

As mentioned in the main text, 
simplified manifold Langevin dynamics do not preserve the target distribution
as they omit third order tensor terms, 
and can provide ineffective proposals for
highly irregular targets (e.g.\@ the banana-shaped distribution in \citep{lan14} or the banana-biscuit-doughnut in \citep{lan16}).
In such cases, the multi-step HMC generalization will also be ineffective  
as the dynamics will soon drift away from the current energy contour, 
and have small acceptance probabilities. 
This consideration motivates a potential future development of 
infinite-dimensional MCMC methods that will incorporate full geometric information (including the third order 
tensor).
The resulting method will be based on the full Riemannian manifold Langevin dynamics (say, on $\mathbb{R}^n$) \citep{girolami11}:
\begin{equation}\label{MLangevin}
\frac{du}{dt} = -\frac{u}{2}+ \frac{g(u)}{2} + \frac{dW^*}{dt}
\end{equation}
where the Brownian motion $W^*$ on the Riemannian manifold 
with metric tensor $G:\mathbb{R}^{n}\mapsto \mathbb{R}^{n\times n}$ has the form \citep{girolami11,chung13}:
\begin{equation}\label{MBrownian}
dW^*(t)_i = |G(u)|^{-\frac12} \sum_j \pa_j [G(u)^{-1}_{ij}|G(u)|^{\frac12}] dt + [\sqrt{G(u)^{-1}} dW]_i
\end{equation}
with $1\le i\le n$,
or the corresponding Lagrangian dynamics \citep{lan14}:
\begin{equation}\label{LD}
\frac{du}{dt} = v\ , \quad \frac{dv}{dt} = -u + g^*(u)\ , 
\end{equation}
where $g^*(u)_i=g(u)_i - \tr[G(u)^{-1} \pa_i G(u)] - \Gamma_{k,l}^i(u)v^kv^l$. 
We have made use of the  Christoffel symbols
$\Gamma_{k,l}^i(u)=\frac12 g^{ij}[\pa_k g_{jl}+\pa_lg_{kj}-\pa_jg_{kl}]$, where $g_{kl}$ denotes the $(k,l)$-th element of $G(u)$.
Combining these dynamics with infinite-dimensional MCMC methodology will require some further research and  is left for future work.
Critically, one will need to carefully investigate the balance between 
improved mixing and the extra computational overheads.

Future work will aim to incorporate alternative dimension reduction techniques such as Likelihood Informed Subspaces \citep[LIS,][]{cui14,cui16} or Active Subspaces \citep[AS,][]{constantine15a,constantine15b}. Fully geometric MCMC can then be employed in the finite dimensional `intrinsic' subspace while its complement can be efficiently explored with relative simple methods like pCN or $\infty$-MALA. This merging of ideas will maybe enable us to make even better use of the geometric structure of the target within the MCMC algorithms.

\section*{Acknowledgement}
We thank Claudia Schillings for her assistance in the development of adjoint codes for the groundwater flow problem and Umberto Villa for his assistance in the development of adjoint codes for the laminar jet problem.
AB is supported by the Leverhulme Trust Prize.
MG, SL and AMS are supported by the EPSRC program grant, Enabling Quantification of Uncertainty in Inverse Problems (EQUIP), EP/K034154/1 and the DARPA funded program Enabling Quantification of Uncertainty in Physical Systems (EQUiPS), contract W911NF-15-2-0121. 
MG is also supported by an EPSRC Established Career Research Fellowship, EP/J016934/2.
PEF is supported by EPSRC grants EP/K030930/1 and  EP/M019721/1, and  a Center of Excellence grant from the Research Council of Norway to the Center for Biomedical
Computing at Simula Research Laboratory.
AMS is also supported by an ONR grant.

\section*{References}

\bibliography{references}

\newpage
\begin{center}
{\Large \bf Appendix: Proofs}
\end{center}
\appendix

\section{Proof of Theorem \ref{thm:infmHMC}}\label{appdx:thm-infmHMC}

\begin{proof}
\begin{itemize}
\item[(i)]
Note that $S^{(i)} = S^{(i-1)}\circ \Psi^{-1}_{\epsilon}$, and that $\Psi_{\epsilon}=\Xi\circ R \circ \Xi$,
where $\Xi$ denotes the first or third map in (\ref{mHDdiscret}) and $R$ the second map (rotation).
Thus, we have the equality
$S^{(i)} = ((S^{(i-1)}\circ \Xi^{-1})\circ R^{-1})
\circ \Xi^{-1}$.
Notice that with this notation $\tilde{S}_0\equiv 
S_0\circ\Xi^{-1}$, so we have 
$G(u,v)= (d\tilde{S}_0/dS_0)(u,v)\equiv
(d(S_0\circ\Xi^{-1})/dS_0)(u,v)$.
 We proceed as follows:
\begin{align*}
&\frac{d((S^{(i-1)}\circ \Xi^{-1})\circ R^{-1})
\circ \Xi^{-1}}{dS_0}(u_i,v_i) = \\
&=\frac{d(((S^{(i-1)}\circ \Xi^{-1})\circ R^{-1})
\circ \Xi^{-1})}{d(S_0\circ\Xi^{-1})}(u_i,v_i)\cdot
\frac{d(S_0\circ\Xi^{-1})}{dS_0}(u_i,v_i)\\
&= \frac{d((S^{(i-1)}\circ \Xi^{-1})\circ R^{-1})}{dS_0}(\Xi^{-1}(u_i,v_i))\cdot G(u_i,v_i)\\
&=\frac{d(S^{(i-1)}\circ \Xi^{-1})}{dS_0}(R^{-1}(\Xi^{-1}(u_i,v_i)))\cdot G(u_i,v_i)\\ &= 
\frac{d(S^{(i-1)}\circ \Xi^{-1})}{d(S_0\circ \Xi^{-1})}(R^{-1}(\Xi^{-1}(u_i,v_i)))\\ &\qquad \qquad \qquad \times 
\frac{d(S_0\circ \Xi^{-1})}{dS_0}(R^{-1}(\Xi^{-1}(u_i,v_i)))
 \cdot G(u_i,v_i)\\ 
& = \frac{dS^{(i-1)}}{dS_0}(u_{i-1},v_{i-1})\cdot 
 G(u_{i-1},v_{i-1}+\tfrac{\epsilon}{2}g(u_{i-1}))
 \cdot G(u_i,v_i)\ . 
\end{align*}
\item[(ii)] It is obtained 
from somewhat cumbersome, but 
straightforward algebraic calculations.
\item[(iii)] Same as (ii).
\item[(iv)]
The proof is similar to the  one of Theorem 3.1 in \cite{beskos2013},
but we include it here for completeness.
The next position, $u^n$, of the Markov chain is 
(for a uniform $U\sim U[0,1]$): 
\begin{align*}
u^n = \delta\,[\,U\le a(u_0,v_0)\,]\,u_I + 
\delta\,[\,U>a(u_0,v_0)\,]\,u_0\ .
\end{align*}
For continuous, bounded $f:\Hi\mapsto \mathbb{R}$,
we take expectations on both sides to obtain:
\begin{align*}
\mathbb{E}\,[\,f(u^n)\,] = 
\mathbb{E}\,[\,a(\Psi_\epsilon^{-I}(u_I,v_I))\,f(u_I)\,] 
- \mathbb{E}\,[\,a(u_0,v_0)\,f(u_0)\,] + \mathbb{E}\,[\,f(u_0)\,]
\ . 
\end{align*}
Thus, it suffices to prove  $\mathbb{E}\,[\,a(\Psi_\epsilon^{-I}(u_I,v_I))\,f(u_I)\,] 
= \mathbb{E}\,[\,a(u_0,v_0)\,f(u_0)\,]$. 
Note now that (we sometimes stress the particular integrators in expectations/integrals  
by showing them explicitly as a subscript of $\mathbb{E}$ when needed):
\begin{align}
\mathbb{E}[\,f(u_I)\,a(\Psi_\epsilon^{-I}(u_I,v_I))\,]& \equiv \mathbb{E}_{\,S^{(I)}}[\,f(u_I)\,a(\Psi_{\epsilon}^{-I}(u_I,v_I))\,]
\nonumber\\
 &= 
  \mathbb{E}_{\,S}[
  \,f(u_I)\,a(\Psi_{\epsilon}^{-I}(u_I,v_I))\,
e^{\Delta H(\Psi_{\epsilon}^{-I}(u_I,v_I))}\,] \nonumber \\
&=  \mathbb{E}_{\,S}[
\,f(u_I)\,(\,1\wedge e^{\Delta H(\Psi_{\epsilon}^{-I}(u_I,v_I))})\,
]\nonumber 
\\
&= \mathbb{E}_{\,S}[\,f(u_I)
(1\wedge e^{\Delta H(\Psi_{\epsilon}^{-I}(u_I,-v_I))} ) \,] \ .\label{eq:aa}
\end{align}

(For the 2nd equation we used the density $dS^{(I)}/dS$ 
we found in (ii) together with the identity in (iii); for the last equation, notice that $(u_I,v_I)$ and $(u_I,-v_I)$ have the same law $S$.)
Now,   due to the symmetry property
$\Psi_{\epsilon}^{I} \circ M \circ \Psi_{\epsilon}^{I} = M $ of the leapfrog operator (we have denoted by $M$ the operator that `flips' the sign of the velocity), 
we have that $ \Psi_{\epsilon}^{-I}\circ M = M\circ \Psi_{\epsilon}^{I}$. Thus, we have:
\begin{align*}
\Delta &H(\Psi_{\epsilon}^{-I}(u_I,-v_I))) = \Delta H(M\circ \Psi_{\epsilon}^{I}(u_I,v_I))) \\
&=  H(M(u_I,v_I)) - H(M\circ \Psi_{\epsilon}^{I}(u_I,v_I))
 \equiv -\Delta H(u_I,v_I)\ ,
\end{align*}
where in the last equation we used the fact that $H\circ M=H$ due to the energy $H$ being quadratic in
the velocity $v$.
Using this in (\ref{eq:aa}), we have  
obtained indeed that $\mathbb{E}\,[\,a(\Psi_\epsilon^{-I}(u_I,v_I))\,f(u_I)\,] 
= \mathbb{E}\,[\,a(u_0,v_0)\,f(u_0)\,]$ 
as required.
%
%
\end{itemize}
\end{proof}

\section{Proof of Corollary \ref{cor:multi-step}}\label{appdx:cor-multi-step}
\begin{proof}
For the given setting of the step-sizes \eqref{step-size-setting},
we first prove the coincidence of the proposals by $\infty$-mHMC and $\infty$-mMALA,
that is,   \eqref{mHDdiscret} reduces to   \eqref{eq:infmMALA}.
Noting that $u_0=u$ and $v_0=\xi\sim \mathcal N(0,K(u))$, with the first equation of \eqref{mHDdiscret} we have:
\begin{equation*}\label{rel1}
v^- = v_0 + \tfrac{\epsilon_1}{2}\,g(u_0) \equiv v\ , 
\end{equation*}
with $v$ as defined in the $\infty$-mMALA proposal in  \eqref{eq:infmMALA}. 
Then, the definition of $\rho$ in \eqref{eq:infMALA} and the setting \eqref{step-size-setting} imply:
\begin{equation*}\label{rel2}
\rho = \frac{1-h/4}{1+h/4} = \cos\eps_2\,;\quad \sqrt{1-\rho^2} = \frac{\sqrt{h}}{1+h/4} = \sin\eps_2\ . 
\end{equation*}
Therefore, it follows from the second equation of \eqref{mHDdiscret},
that the proposal, say $u'$, of $\infty$-mHMC for one leapfrog step is equal 
to: 
\begin{equation*}
u' = u_{\eps_2} = u_0\cos\eps_2+v^-\sin\eps_2 \equiv  \rho\,u + \sqrt{1-\rho^2} \,v \ , 
\end{equation*}
with the term on the right hand side being the proposal from $\infty$-mMALA.
Since the proposals coincide, the acceptance probabilities will also be the same, as they both apply the Metropolis-Hastings ratio.

\end{proof}

\end{document}